\documentclass[reqno]{amsart}

\usepackage{amsmath,amssymb,amsthm,amsfonts}
\usepackage{ulem} 
\usepackage{color}
\usepackage{hyperref}
\usepackage{a4wide}
\usepackage{mathrsfs}
\usepackage{tikz}
\usetikzlibrary{fit,trees,quotes,graphs,positioning,calc,arrows.meta,bending,patterns,topaths,math,decorations.markings,intersections,through}
\usetikzlibrary{shapes.geometric} 

\makeatletter
\tikzset{
  use path for main/.code={%
    \tikz@addmode{%
      \expandafter\pgfsyssoftpath@setcurrentpath\csname tikz@intersect@path@name@#1\endcsname
    }%
  },
  use path for actions/.code={%
    \expandafter\def\expandafter\tikz@preactions\expandafter{\tikz@preactions\expandafter\let\expandafter\tikz@actions@path\csname tikz@intersect@path@name@#1\endcsname}%
  },
  use path/.style={%
    use path for main=#1,
    use path for actions=#1,
  }
}
\makeatother

\numberwithin{equation}{section}

\newtheorem{theorem}{Theorem}[section]
\newtheorem{lemma}[theorem]{Lemma}
\newtheorem{prop}[theorem] {Proposition}
\newtheorem{cor}[theorem]  {Corollary}
\newtheorem{definition}[theorem] {Definition}

\theoremstyle{definition}

\theoremstyle{remark}
\newtheorem{remark}[theorem]{Remark}
\newtheorem{example}[theorem]{Example}

\newcommand{\e}{\mathrm{e}} 
\newcommand{\N}{\mathbb{N}}
\newcommand{\R}{\mathbb{R}}

\newcommand{\C}{\mathbb{C}}

\newcommand{\dd}{\mathrm{d}} 

\newcommand{\vect}[1]{\boldsymbol{#1}}

\newcommand{\col}{C} 

\newcommand{\be}{\begin{equation}}
\newcommand{\ee}{\end{equation}}
\newcommand{\ba}{\begin{equation} \begin{aligned}}
\newcommand{\ea}{\end{aligned}\end{equation}}
\newcommand{\bes}{\begin{equation*}}
\newcommand{\ees}{\end{equation*}}

\def\1{{\mathchoice {1\mskip-4mu\mathrm l}      
{1\mskip-4mu\mathrm l}
{1\mskip-4.5mu\mathrm l} {1\mskip-5mu\mathrm l}}}
\newcommand{\heap}[2]{\genfrac{}{}{0pt}{}{#1}{#2}}

\begin{document}

\title{Lagrange inversion  and combinatorial species with uncountable color palette}
\author{Sabine Jansen}
\address{Mathematisches Institut, Ludwig-Maximilians-Universit{\"a}t, 80333 M{\"u}nchen,  Germany}
\email{jansen@math.lmu.de} 
\author{Tobias Kuna}
\address{Department of Mathematics and Statistics,
University of Reading, Reading RG6 6AX, UK}
\email{t.kuna@reading.ac.uk}
\author{ Dimitrios Tsagkarogiannis}
\address{Dipartimento di Ingegneria e Scienze dell'Informazione e Matematica, Universit\`a degli Studi dell'Aquila, 67100 L'Aquila, Italy}
\email{dimitrios.tsagkarogiannis@univaq.it}

\date{\today}
\begin{abstract} 
We prove a multivariate Lagrange-Good formula for functionals of uncountably many variables
and investigate its relation with inversion formulas using trees.
We clarify the cancellations that take place between the two aforementioned formulas and draw connections with similar approaches in a range of applications.
\\

 AMS math classification: 32A05, 47J07, 05A15, 82B21, 81T15 ]
\end{abstract}

\maketitle

\tableofcontents

\section{Introduction} 

Consider a power series of the form 
\be\label{iteration}
	\rho(z) =z\, \e^{- A(z)},\quad A(z) =\sum_{n=1}^\infty \frac{a_n}{n!}\, z^n.
\ee
In a previous paper \cite{jansen-kuna-tsagkaro2019virialinversion}, we have expressed the inverse map $\zeta(\nu)$ (that is $\rho(z) = \nu$ if and only if $z = \zeta(\nu)$) with the help of a tree generating functional.  The main focus is to formulate these results for \emph{uncountably} many colors. An uncountable color palette is completely natural if one studies thermodynamic functionals for inhomogeneous systems in order to derive variational or PDE formulations in different areas, for example in classical density function theory, liquid crystals, heterogenous materials, colloid systems, system of molecules with various shapes or other internal degrees of freedom.  This is because ``colors'' correspond to positions or other continuous degrees of freedom (e.g. orientations). 

If $\rho$ and $z$ are taking only non-negative values, as in the above mentioned applications, then the form \eqref{iteration} is natural. 

For a single complex variable $z\in \C$, the inverse reads 
\be
	\zeta(\nu) = \nu T^\circ(\nu)
\ee
where $T^\circ(\nu)$ solves 
\be
	T^\circ(\nu) = \exp\Biggl( \sum_{n=1}^\infty \frac{a_n}{n!}\, T^\circ(\nu)^n\Biggr)
\ee
and one recognizes the generating function for weighted rooted trees whose root is a ghost (that is,  the root does not come with powers of $\nu$ in the generating function). 

On the other hand, as is well-known, the coefficients of $\nu$ in the inverse map $\zeta(\nu)$ (and any functions of $z=\zeta(\nu)$) can be expressed in terms of the Lagrange inversion formula. Formally, 
\be \label{eq:single-lagrange}
\begin{aligned}
	 \left[\nu^n\right]  f(\zeta(\nu)) & = \frac{1}{2\pi \mathrm i}\oint \frac{f(\zeta(\nu))}{\nu^n} \frac{\dd \nu}{\nu} \\
	&	=\frac{1}{2\pi \mathrm i}\oint \frac{f(\zeta\bigl(\rho(z)\bigr))}{\rho(z)^n} \rho'(z)\, \frac{\dd z}{\rho(z)} \\
		& = \frac{1}{2\pi \mathrm i}\oint \frac{f(z)}{\rho(z)^n} \rho'(z)\, \frac{\dd z}{\rho(z)} \\
		& = \left[ z^n \right] \frac{f(z)}{[\rho(z)/z]^{n+1}}\, \rho'(z).
\end{aligned} 
\ee
An additional integration by parts yields the more frequently encountered form of the Lagrange inversion formula 
\be\label{withoutderivative}
	\left[\nu^n\right]  f(\zeta(\nu)) =  \frac{1}{n} \left[ z^{n-1} \right] \frac{f(z)}{[\rho(z)/z]^{n}},
\ee
see e.g. \cite[Appendix A6]{flajolet-sedgewick2009book} or the recent survey \cite{gessel2016survey}. The multivariate case is similar, with the determinant of a Jacobi matrix instead of the derivative $\rho'(z)$.  There is a variety of multivariate forms, see \cite{gessel1987}, \cite{good1960} and \cite{percus1964lagrange}, or  \cite{barnabei1985}, \cite{ehrenborg-mendez1994} and \cite{jttu2014} for inifinitely (countably) many variables.

Despite the large literature on Lagrange-Good inversion and the increasing interest it attracts from combinatorialists, to the best of our knowledge  no analogue for uncountably many variables has been considered up to now. An exception is~\cite{mendez-nava1993}, where combinatorial identities are generalized to an arbitrary number of variables, however the concept of summability by Mend{\'e}z and Nava  is a restriction which would not allow to treat the above mentioned applications.
The first aim of the present note is to fill this gap: we propose a multivariate Lagrange-Good formula for functionals of uncountably many variables (Theorem~\ref{thm:lagrange-good}).

The second aim is to clarify the relation with the tree formula from Proposition 2.6 in  \cite{jansen-kuna-tsagkaro2019virialinversion}. Just as Gessel's proof of the Lagrange-Good inversion formula for finitely many variables, our proof starts from a representation of the inverse in terms of trees. In contrast, determinants are associated with digraphs that may have cycles; equality arises because of cancellations as clarified in Proposition~\ref{prop:magicformula}. 
There is actually more to this: when interested in proving asymptotic formulas and checking
the validity (absolute convergence) of the power series, the tree formula is easier to handle.
In fact, in \cite{jttu2014} we had to show that the determinant was bounded and then 
in \cite{jansen-kuna-tsagkaro2019virialinversion} we realized that the determinant is not actually
there due to the cancellations and therefore we deduced better bounds with less effort. The observation that determinants might be a hindrance to asymptotic analysis is also behind determinant-free, so-called \emph{arborescent forms} of Lagrange-Good, see \cite{bender-richmond1998, bousquet-chauve-labelle-leroux2003, goulden-kulkarni1997}.

It is noteworthy that  Good's original motivation \cite{good1960,good1965} for generalizing Lagrange's inversion formula was the enumeration of various kinds of trees, thus starting from the functional equation satisfied by combinatorial generating functions. The antipodal view point is to start from the inversion problem, consisting in solving a given functional equation; then one has two options: either apply one of the versions of the Lagrange-Good inversion formula, or derive an expression for the solution directly with trees. 
Solving inversion problems with different types of trees is common practice in many areas: Butcher series in numerics, Gallavotti trees in RG group, Lindstedt series in KAM theory \cite{gallavotti2012perturbation}, algebra \cite{bass-connell-wright1982} and \cite{wright1989treeformula}, see Section~\ref{sec:discussion} for more details and a discussion of the relation to the approach taken in this paper. 

The article is organized as follows. Section~\ref{sec:prelim} recalls the multi-variate formulas that we seek to generalize and introduces formal power series associated with uncountable spaces. In Section~\ref{sec:main} we set up the problem and present the main result, which is then proven in Section~\ref{sec:proofs}. We draw connections to other inversion formulas based on trees in Section~\ref{sec:discussion}. Appendix~\ref{ApFormal} recalls some formulas for formal power series. For our proof, we propose a formalism of colored species with uncountable color space, which allows us to adapt the proof given by Gessel (1987) \cite{gessel1987} for finitely many variables to the case of uncountably many variables; this formalism is described in detail in Appendix~\ref{app:coloredspecies}. 

\section{Preliminaries} \label{sec:prelim}

\subsection{Multivariate Lagrange inversion} 

Here we recall the form of the Lagrange inversion formula for formal power series in finitely many variables $z_1,\ldots, z_d\in \mathbb C$ that we seek to generalize. For $\vect n= (n_1,\ldots, n_d) \in \N_0^d$ we write $\vect z^{\vect n} = z_1^{n_1}\cdots z_d^{n_d}$ and $[\vect z^{\vect n}] F(\vect z)$ denotes the coefficient of the monomial $\vect z^{\vect n}$ in the series $F$, i.e., if $F(\vect z)= \sum_{\vect n} f_{\vect n} \vect z^{\vect n}$, then $[\vect z^{\vect n}] F(\vect z) = f_{\vect n}$.

Suppose we are given a family $(A_i(z_1,\ldots, z_d))_{i=1,\ldots,d}$ of formal power series whose coefficient of order zero vanishes, $A_i(\vect 0) =0$.  Define additional formal power series\footnote{The choice of letters $\rho_i$ and $z_i$ as well as the exponential form of the map are motivated by applications in statistical mechanics~\cite{jttu2014}, where the $z_i$'s and $\rho_i$'s correspond to activity and density variables and the index $i$ may refer to the type of a particle or a discrete set of locations on a lattice. The exponential from in Eq.~\eqref{rho} will be crucial in the following.} 
 $\rho_1,\ldots, \rho_d$ by 
\be\label{rho}
	\rho_i(z_1,\ldots, z_d):= z_i \exp \bigl( - A_i(z_1,\ldots, z_d)\bigr). 
\ee
There is a uniquely defined family of power series $(\zeta_1,\ldots, \zeta_d)$ such that 
\be
	\zeta_i\bigl( \rho_1(\vect z),\ldots, \rho_d(\vect z)\bigr) = z_i,\quad (i=1,\ldots,d)
\ee
as an equality of formal power series, furthermore $(\zeta_1,\ldots, \zeta_d)$ also satisfies 
\be\label{rho2k}
	\rho_i\bigl( \zeta_1(\vect \nu),\ldots, \zeta_d(\vect \nu) \bigr) = \nu_i \quad (i=1,\ldots,d)
\ee
as an equality of formal power series in the variables $\nu_1,\ldots, \nu_d$. 
Let $\Phi(z_1,\ldots,z_d)$ be yet another formal power series. Then 
\be \label{eq:finite-lagrange1}
	[\vect \nu^{\vect n}]\, \Phi\bigl( \zeta_1(\vect \nu),\ldots, \zeta_d(\vect \nu)\bigr) = [\vect z^{\vect n}]\,  \left\{\Phi(\vect z)\exp\left(\sum_{k=1}^d n_k A_k(\vect z)\right) 
\det\left(\delta_{ij}-z_i\frac{\partial}{\partial z_j}A_i(\vect z)\right)_{1\leq i,j\leq d}\right\}.
\ee
Eq.~\eqref{eq:finite-lagrange1} follows from~\cite[Eq.~(4.5)]{gessel1987}, see also Theorem 8(b) in \cite[Chapter 3.2]{bergeron-labelle-leroux1998book}. 
Variants for countably many variables are available~\cite{barnabei1985,ehrenborg-mendez1994}. 

 Specializing to $\Phi(\vect z) = z_k$ with $k\in \{1,\ldots, d\}$, we obtain 
\be \label{eq:finite-lagrange2}
	[\vect \nu^{\vect n}]\, \zeta_k(\vect \nu) = [\vect z^{\vect n}]\,  \left\{z_k\, \exp\left(\sum_{k=1}^d n_k A_k(\vect z)\right) 
\det\left(\delta_{ij}-z_i\frac{\partial}{\partial z_j}A_i(\vect z)\right)_{1\leq i,j\leq d}\right\},
\ee
compare~\eqref{eq:single-lagrange} for $d=1$. Conversely, Eq.~\eqref{eq:finite-lagrange1} easily follows from Eq.~\eqref{eq:finite-lagrange2} as well.
Eq.~\eqref{eq:finite-lagrange2} allows us to express the coefficients of the unknown reverse series $\zeta_i(\vect \nu)$ in terms of coefficients of the known series $A_i(\vect z)$ and functions thereof. Eq.~\eqref{eq:finite-lagrange2} should be contrasted with  tree formulas derived from the following functional equation (combining \eqref{rho} and \eqref{rho2k} )
\be \label{eq:fifu}
	\zeta_k(\vect \nu) = \nu_k \, \exp \Bigl( A_k\bigl( \zeta_1(\vect \nu),\ldots, \zeta_d(\vect \nu)\bigr)\Bigr),
\ee
see also Section~\ref{sec:discussion}.

Our goal is to generalize Eqs.~\eqref{eq:finite-lagrange1} and~\eqref{eq:finite-lagrange2} to a situation where the finite set $\{1,\ldots, d\}$ is replaced with a possibly uncountable space $\mathbb X$, and to provide adequate combinatorial interpretations. The tree formula for the inverse in uncountable color space was already proven in~\cite{jansen-kuna-tsagkaro2019virialinversion}, it is recalled in Section~\ref{sec:treeinverse}. 

\subsection{Formal power series} 

 Let us briefly motivate the definition of formal power series adopted in~\cite[Appendix~A]{jansen-kuna-tsagkaro2019virialinversion}. A formal power series in finitely many variables $z_1,\ldots, z_d$ may be written as 
\be \label{eq:power1}
	F(z_1,\ldots, z_d) = \sum_{\vect n\in \N_0^d} \frac{\vect z^{\vect n}}{\vect n!}\, a(\vect n) 
\ee
for some suitable family of coefficients $a(\vect n)$, where $\vect n ! =n_1! \ldots n_d!$ but also as 
\be \label{eq:power2}
	F(z_1,\ldots, z_d) = \sum_{n=0}^\infty \frac{1}{n!} \sum_{i_1,\ldots, i_n=1}^d f_n(i_1,\ldots, i_n) z_{i_1}\cdots z_{i_n}
\ee
with coefficients $f_n(i_1,\ldots, i_n)$ that are invariant under permutation of the argument, where $ a(\vect n) = f_n(i_1,\ldots, i_n)$, whenever $ \#\{ i_1, \ldots , i_d \, :\, i_n =j\} = n_j$ for all $j\in \{1 , \ldots ,d\}$ and $n = n_1 + \ldots + n_d$.
Eq.~\eqref{eq:power2} is less elegant than the multiindex formula~\eqref{eq:power1} but redeems itself by a straightforward  generalization to uncountable spaces. (Other benefits of Eq.~\eqref{eq:power2} are discussed by Abdesselam in the context of Feynman diagrams and tensors~\cite[Section III]{abdesselam2003feynman}.) If $\{1,\ldots, d\}$ is replaced with a measurable space $(\mathbb X,\mathcal X)$ it is natural to replace sums with integrals and switch from a vector $(z_1,\ldots, z_d)$ to a measure $z(\dd x)$. 
Accordingly the power series we are interested in are of the form 
\be \label{eq:powmeas}
	F(z) = f_0 + \sum_{n=1}^\infty\frac{1}{n!}\int_{\mathbb X^n} f_n(x_1,\ldots, x_n) z(\dd x_1) \cdots z(\dd x_n),
\ee
the coefficients consist of a scalar $f_0\in \C$ and symmetric measurable functions $f_n: \mathbb X^n\to \C$, $n\in \N$. As usual for formal power series, we do not want to deal with questions of convergence and downgrade~\eqref{eq:powmeas} to a mnemotechnic notation for the sequence $(f_n)_{n\in \N_0}$. We also want to define function- and measure-valued formal power series $F(q;z)$, $K(\dd q;z)$. 

\begin{definition} \label{def:fps} 
	Let $(\mathbb X, \mathcal X)$ be a measurable space. 
	\begin{enumerate} 
		\item [(a)]	A (scalar) \emph{formal power series} on $\mathbb X$ is a family $(f_n)_{n\in \N_0}$ consisting of a scalar $f_0\in \mathbb C$ and symmetric measurable functions $f_n:\mathbb X^n\to \C$.
		\item [(b)] A \emph{function-valued formal power series} is a family $(f_n)_{n\in \N_0}$ of measurable maps $f_n:\mathbb X\times \mathbb X^n\to \C$ such that $(q,x_1,\ldots, x_n)\mapsto f_n(q;x_1,\ldots, x_n)$ is symmetric in the $x_j$ variables. 
		\item [(c)] A \emph{measure-valued formal power series} is a family $(k_n(\dd q;x_1,\ldots, x_n))_{n\in \N_0}$  consisting of a measure $k_0(\dd q)$ on $\mathbb X$ and kernels $k_n: \mathcal X\times \mathbb X^n\to \R_+$ such that for each $B\in \mathcal X$, the map $k_n(B;\cdot)$ is symmetric in the $x_j$ variables. 
	\end{enumerate} 
\end{definition} 

\noindent Definition~\ref{def:fps}(b) and (c) are formulated for functions of a single variable $q\in \mathbb X$ and non-negative measures, they extend in a straightforward way to functions of several variables or complex measures.
Standard operations such as sums, products, etc. are defined in Appendix~\ref{ApFormal}. The operations turn the set of scalar formal power series into an algebra that corresponds to the algebra of symmetric functions from Ruelle~\cite[Chapter 4.4]{ruelle1969book}.

\begin{example}[Exponential] Let $\varphi$ be a non-negative measurable function on $\mathbb{X}$, then the power series associated to the following exponential is 
\[
 \e^{\int_{\mathbb{X}} \varphi(x) z(\mathrm d x)} =1+ \sum_{n=1}^\infty \frac{1}{n!} \int_{\mathbb{X}^n} \varphi(x_1) \ldots \varphi(x_n) z^n (\dd \vect x).
 \]
\end{example}

\begin{example} [Monomials] \label{ex:monomial} 
	Let $m\in \N$ and $q_1,\ldots, q_m\in \mathbb X$. Let $\delta_{p,q}$ be the Kronecker delta, equal to $1$ if $p = q$ and $0$ if $p\neq q$. Then for every measure $z$ on $\mathbb X$ we have 
	\be 
		z(\{q_1\}) \cdots z(\{q_m\}) = \int_{\mathbb X^m} \delta_{q_1,x_1}\cdots \delta_{q_m,x_m} \, z^m(\dd \vect x) = \frac{1}{m!}\int \sum_{\sigma\in \mathfrak {S}_m} \prod_{i=1}^m \delta_{q_i, x_{\sigma (i)}} \, z^m(\dd \vect x),
	\ee
	which is of the form~\eqref{eq:powmeas} with $f_n \equiv 0$ for all $n\neq m$. 
\end{example} 

\begin{example}[The measure $z(\dd q)$] \label{ex:z}
	Let $z(\dd q)$ be a measure on $(\mathbb X, \mathcal X)$. Then for all $B\in \mathcal X$, 
	\be
		z(B) = \int_\mathbb X \1_B(x) z(\dd x) = \int_\mathbb X \delta_x(B) z(\dd x) = \int_\mathbb X k_1(B;x) z(\dd x),
	\ee	
	with kernel $k_1(B;x)=\delta_{x}(B)$, i.e., $k_1(\dd q;x)$ is the Dirac measure at $x$. Accordingly we may view $z(\dd q)$ as a measure-valued formal power series in the sense of Definition~\ref{def:fps}(c), with $k_n\equiv 0$ for all $n\neq 1$ and  $k_1(\dd q; x) = \delta_x(\dd q)$. The measure $z(\dd q)$ replaces the set of monomials $(z_i)_{i=1,\ldots, d}$ that appear naturally for power series of finitely many variables, with $z(B)$ the analogue of $\sum_{i\in B} z_i$.   
\end{example}

\subsection{Variational derivatives and extraction of coefficients} 

Just as for usual power series, coefficients can be extracted by taking derivatives at the origin. In our context, the correct notion of derivative is a variational derivative defined as follows. 

\begin{definition} \label{def:varderivative}
	Let $F(z) = f_0 + \sum_{n=1}^\infty \frac{1}{n!} \int_{\mathbb X^n} f_n(x_1,\ldots, x_n) z^n(\dd \vect x)$ be a formal power series, i.e., $(f_n)_{n\in \N_0}$ is a family of symmetric functions as in Definition~\ref{def:fps}(a). The variational derivative of order $k$ is the function-valued formal power series with coefficients 
	\be
		\Bigl(\frac{\delta^k f}{\delta z^k}\Bigr)_n(q_1,\ldots,q_k;x_1,\ldots,x_n) := 
			f_{k+n}(q_1,\ldots, q_k,x_1,\ldots, x_n).
	\ee
\end{definition} 

\noindent Thus 
\be
	 \frac{\delta^k f}{\delta z^k}(q_1,\ldots,q_k;z) =f_k(q_1,\ldots,q_k) +  \sum_{n=1}^\infty \frac{1}{n!} \int_{\mathbb X^n} f_{k+n}(q_1,\ldots, q_m,x_1,\ldots, x_n) z^n(\dd \vect x). 
\ee
In particular, $f_k(q_1,\ldots, q_k;z)$ is equal to the term of order zero in the formal power series $\frac{\delta^k f}{\delta z^k}(\vect q;z)$. 
Below we often use the  notation
\be
	\Bigl(\frac{\delta^k}{\delta z(q_1)\cdots \delta z(q_k)} f\Bigr)(z) 
	= \frac{\delta^k f}{\delta z^k}(q_1,\ldots, q_k;z).
\ee
Definition~\ref{def:varderivative} is motivated by the following formal computation. For small $t\in \R$ and another measure $\mu$ on $\mathbb X$, we have
\be
\begin{aligned}
	F(z+ t\mu) & = F_0  + \sum_{n=1}^\infty \frac{1}{n!}\sum_{k=0}^n t^k \sum_{\substack{J\subset [n] \\ \#J =k}} \int_{\mathbb X^n} F_n(x_1,\ldots, x_n) z^{n-k}(\dd \vect x_{[n]\setminus J}) \mu^k(\dd \vect x_J) \\
	& = F_0  + \sum_{n=1}^\infty \frac{1}{n!}\sum_{k=0}^n \binom{n}{k} t^k  \int_{\mathbb X^n} F_n(q_1,\ldots,q_{n-k}, y_1,\ldots, y_{k}) z^{n-k}(\dd \vect q) \mu^k(\dd \vect y) \\
	& = F(z) + \sum_{k=1}^\infty \frac{t^k}{k!} \int_{\mathbb X^k} \frac{\delta^k F}{\delta z^k}(q_1,\ldots, q_k;z) \mu^k(\dd \vect x)
\end{aligned}
\ee
and 
\be
	\left. \frac{\dd^k}{\dd t^k} F(z+ t\mu) \right|_{t=0} = \int_{\mathbb X^k} \frac{\delta^k F}{\delta z^k}(q_1,\ldots, q_k;z) \mu^k(\dd \vect x).
\ee

\subsection{Fredholm determinant} \label{sec:rmk:nice}
Let $K:\mathbb X\times \mathbb X\to \R$ be a kernel and $\mathbb K$ the associated integral operator in $L^2(\mathbb X, \mathcal X, z(\dd x))$, given by
\[
	(\mathbb K \varphi)(q') = \int_\mathbb X K(q',q) z(\dd q). 
\] 
For sufficiently regular kernels $K$, the Fredholm determinant 
$\det( \mathrm{Id} - \mathbb K)$ is
\be \label{eq:fredformal1aa}
	\det( \mathrm{Id} - \mathbb K) = 1 + \sum_{n=1}^\infty \frac{(-1)^n}{n!}\int_{\mathbb X^n} \det\left( \left( K(q_i,q_j) \right)_{i,j=1, \ldots ,n} \right) z(\dd q_1) \ldots z(\dd q_n),
\ee
see e.g. Subsection~3.11 in \cite{Simon2015Operator}. The right-hand side of~\eqref{eq:fredformal1aa} is always well-defined as a formal power series in $z$, without any regularity assumptions on the kernel. Accordingly we adopt~\eqref{eq:fredformal1aa} as a definition of the Fredholm determinant on the level of formal power series. 

The definition is easily extended to kernels $K_z$ and associated operators $\mathbb K_z$ that are themselves given by formal power series, as in Eq.~\eqref{eq:aqdev} below. That is, suppose we are given a family of measurable functions $k_0:\mathbb X\times \mathbb X\to \R$ and $k_n : {\mathbb X} \times \mathbb X \times {\mathbb X}^n \rightarrow \mathbb{C}$ such that $(q,q';x_1 ,\ldots x_n) \mapsto k_n(q,q';x_1 ,\ldots x_n)$ is symmetric in the $x_i$ variables, and define
\[
	K_z(q',q) = k_0(q',q) + \sum_{n=1}^\infty \frac{1}{n!}\int_{\mathbb X^n} k_n(q',q;x_1,\ldots,x_n) z^n(\mathrm d \vect x). 
\] 
The determinant $\mathrm{det}\, ((K_z(q_i,q_j))_{i,j=1,\ldots,n})$ is a combination of products of power series, the integral of formal power series is defined using~\eqref{eq:formal-integral}. Therefore, the $n\times n$ determinants and integrals in~\eqref{eq:fredformal1aa} stay well-defined as a formal power series with $K_z$ instead of $K$. For the $n$-th summand on the right-hand side of~\eqref{eq:fredformal1aa}, 
the first non-zero term in the power series expansion has degree $n$, because of the ``integration'' with $z^n(\dd \vect x)$. Hence the contributions for the term of degree $m \in \mathbb{N}$ in the formal power series of $K_z$ comes from summands on the right hand side for $n \leq m$. Therefore, the coefficient of the degree $m$ is a finite sum of finite products of the functions $k_n$ and hence the infinite series on the right hand side is rigorously defined  as a formal power series.

We conclude with a remark that may be helpful for readers which are not too keen on working with Fredholm determinants.
\begin{remark}\label{remfredfini} 
One can replace the Fredholm determinant  with determinants of finite matrices. 
Recall from Example~\ref{ex:monomial} how to interpret $z(\{q\})$ as a formal power series.

	Let $n\in \N$ and $(q_1,\ldots,q_n)\in \mathbb X^n$. Set $Q=\{q_i:\, i=1,\ldots,n\}$. Then the 
	$n$-th coefficient of the Fredholm determinant $\det (\mathrm{Id}- \mathbb K_z)$, evaluated at $(q_1,\ldots,q_n)$, is equal to the $n$-th coefficient at $(q_1,\ldots,q_n)$ of the $(\#Q)\times (\#Q)$-matrix 
\be \label{eq:smalldeterminanta} 
	\mathrm{det}\left(\Bigl( \delta_{q,q'} - z(\{q\}) K_z(q',q)\Bigr)_{q,q'\in Q}\right).
\ee
One can replace $Q$ by any $Q' \supset Q$ without altering the coefficient

If one additionally wants to avoid the use of measures for $z$, because one is either only interested in densities or maybe even in generalized functions, then in order to compute the $n$-th. coefficient it is sufficient to consider the following determinant
\be
\int_{\mathcal{X}}\mathrm{det}\left(\Bigl( \delta(x_j - x_j) - z(x_j)  K_z(x_i,x_j)\Bigr)_{i,j = 1, \ldots , n}\right) m(\dd x_1) \ldots m(\dd x_n),
\ee
where $m$ is a reference measure on $\mathbb{X}$, for example typically  the Lebesgue measure. 

For both cases, more details can be found at the end of Appendix~\ref{ApFormal}.
The use of the Fredholm determinant gives the most natural connection to the combinatorics. All other interpretations will give rise to the same result as the coefficients of the associated formal power series are unchanged.
\end{remark}

\section{Main results} \label{sec:main} 

Let $(\mathbb X, \mathcal X)$ be a measurable space and 
\be
	A(q;z):= A_0(q) + \sum_{n=1}^\infty \frac{1}{n!}\int_{\mathbb X^n} A_n(q;x_1,\ldots, x_n) z(\dd x_1) \cdots z(\dd x_n)
\ee
a function-valued formal power series in the sense of Definition~\ref{def:fps}(b). Thus each $A_n: \mathbb X\times \mathbb X^n\to \C$ is a measurable function that is symmetric in the $x_j$-variables. Define a measure-valued formal power series $\rho(\dd q;z)$ by 
\be
	\rho[z](\dd q) \equiv	\rho(\dd q;z) = z(\dd q) \exp\bigl( - A(q;z)\bigr). 
\ee
The coefficients of the power series on the right-hand side are defined rigorously by Eqs.~\eqref{eq:formal-exponential} and~\eqref{eq:formal-radon}. 
We would like to determine the inverse power series $\zeta[\nu](\dd q)= \zeta(\dd q;\nu)$, that is
\be \label{eq:inverse}
	(\zeta\circ\rho)(\dd q;z) = z(\dd q),\quad (\rho\circ \zeta)(\dd q;\nu) = \nu(\dd q)
\ee
with the composition defined by~\eqref{eq:fmc}. 
In a previous article \cite{jansen-kuna-tsagkaro2019virialinversion} we have proven that the inversion is always possible on the level of formal power series, and  we gave sufficient conditions for  the absolute convergence of the involved power series. Precisely, concerning the formal inverse, we have proven that there is a unique family of formal power series $(T_q^\circ)_{q\in \mathbb X}$ that solves the fixed point equation 
\be \label{tree-eq}\tag{$\mathsf{FP}$}
	T_q^\circ(\nu) = \exp\Biggl( \sum_{n=1}^\infty \frac{1}{n!} \int_{\mathbb X^n} A_n(q;x_1,\ldots, x_n) T_{x_1}^\circ(\nu)\cdots T_{x_n}^\circ(\nu) \nu(\dd x_1)\cdots \nu(\dd x_n)\Biggr),
\ee 
compare \cite[Theorem~3.2.2]{bergeron-labelle-leroux1998book} for finitely many variables. Moreover the measure-valued formal power series 
\be \label{eq:zetadef}
	\zeta[\nu] (\dd q)  := T_q^\circ (\nu) \nu(\dd q)
\ee
(see again~\eqref{eq:formal-radon}) satisfies
Eq.~\eqref{eq:inverse}.

Further we have shown [JKT2019, Proposition 2.6] that the power series $T_q^\circ$ is the generating function for rooted weighted trees whose root has color $q$ and is a ghost (it does not come with powers of $z$ in the generating functional). The tree formula for the inverse is recalled in detail in Section~\ref{sec:treeinverse}. Corollary~\ref{cor:inverse}  below provides an alternative representation as the coefficient of another power series, generalizing the multivariate Lagrange inversion formula~\eqref{eq:finite-lagrange2}. 

First, however, we generalize Eq.~\eqref{eq:finite-lagrange1} to uncountably many colors. 
 The determinant in Eq.~\eqref{eq:finite-lagrange1} is replaced with a Fredholm determinant. 
Define the kernel 
\[
	K_z(q',q) = \frac{\delta}{\delta z(q')} A(q;z).
\] 
 The variational derivative has been introduced in Definition~\ref{def:varderivative}, see also  Eq.~\eqref{fvardev}. In particular, 
\be \label{eq:aqdev}
	K_z(q',q)  = A_1(q;q') + \sum_{m=1}^\infty \frac{1}{m!} \int_{\mathbb X^m} A_{m+1}(q;q',x_1,\ldots, x_m) z^m(\dd \vect x).
\ee
Consider the formal operator 
\be \label{eq:operatordef}
	\bigl(\mathbb A_z \varphi \bigr)(q'):= \int_{\mathbb X} z(\dd q) \varphi(q) \frac{\delta}{\delta z(q')} A(q;z) = \int_\mathbb X K_z(q',q) \varphi(q) z(\dd q)
\ee
which we may view as a (formal) integral operator in $L^2(\mathbb X,\mathcal X,z(\dd q))$ with kernel $K_z(q',q)$. The Fredholm determinant $\mathrm{det}(\mathrm{Id} - \mathbb A_z)$ is defined, as a formal power series, as in Section~\ref{sec:rmk:nice}, see also Eq.~\eqref{eq:fredformal} below.

\begin{theorem} [Lagrange-Good inversion] \label{thm:lagrange-good}
	Let $\Phi(z) = \Phi_0+ \sum_{n=1}^\infty\frac{1}{n!}\int_{\mathbb X^n} \Phi_n(\vect{x}) z^n(\dd \vect{x})$ be a formal power series. Define a formal power series $\Psi$ (using ~\eqref{eq:fmc}) by  
	\be
 \Phi(\zeta(\nu)) = \Psi(\nu) = \Psi_0 +  \sum_{n=1}^\infty\frac{1}{n!}\int_{\mathbb X^n} \Psi_n(\vect{q}) \nu^n(\dd\vect{q}).
	\ee
	Then, for all $n\in \N$ and $(q_1,\ldots,q_n) \in \mathbb X^n$, $ \Psi_n(q_1,\ldots,q_n)$ is equal to the term of order zero in the formal power series 
	\be\label{rhsQ}
		\frac{\delta^n}{\delta z(q_1)\cdots \delta z(q_n)} \left\{ \Phi(z) \exp\left( \sum_{i=1}^n A(q_i;z) \right)  \det \left( \mathrm{Id} 
		- \mathbb A_z \right)  
		\right\}.
	\ee 
\end{theorem}

\noindent By the definitions adopted in Section~\ref{sec:rmk:nice}, the Fredholm determinant in~\eqref{rhsQ} is given by 
\be \label{eq:fredformal}
 \det \left( \mathrm{Id}- \mathbb A_z \right) 
		  = 1 + \sum_{n=1}^\infty \frac{(-1)^n}{n!}\int_{\mathbb X^n}  \det\left( \left( \frac{\delta}{\delta z(x_i)} A(x_j;z) \right)_{i,j=1, \ldots ,n} \right) z(\dd x_1) \ldots z(\dd x_n), 
\ee
we will sometimes use the heuristic notation
\[
	\det(\mathrm{Id}- \mathbb A_z) = \det \left( \mathrm{Id} 
		- z(\dd q)\frac{\delta}{\delta z(q')} A(q;z) \right).
\]

\begin{remark}\label{rmk:nice}
	Following Remark~~\ref{remfredfini}, the Fredholm determinant in Theorem~\ref{thm:lagrange-good} can be replaced by usual determinants in two ways:
	\begin{enumerate}
\item Using restrictions, we get that	
    \be\label{rhs}
       \det \left( \left( \delta_{qq'} - z(\{q\})\frac{\delta}{\delta z(q')} A(q;z)\right)_{q,q'\in Q}\right) ,
	\ee
	where
	\be\label{Q}
		Q:=\{q_i \mid i =1,\ldots, n\}.
	\ee
The replacement is also possible as well if we use instead of $Q$ any bigger set $Q'\supset Q$. 
\item If $z$ has a density with respect to a reference measure $m$, which we denote also by $z$, or if $z$ is a generalized function one can replace the Fredholm determinant by
\be
\int_{\mathcal{X}}\mathrm{det}\left(\Bigl( \delta(x_j - x_j) - z(x_i) \frac{\delta}{\delta z(x_j)} A(x_i;z) \Bigr)_{i,j = 1, \ldots , n}\right) m(\dd x_1) \ldots m(\dd x_n),
\ee
where the expression is well-defined for nice enough $A$; indeed see \eqref{eq:fredformal}.
\end{enumerate}

The usual determinant may feel more elementary, however the Fredholm determinant is defined as a series of usual determinants anyhow. If one wishes to apply the theorem in a functional-analytic context  the Fredholm determinant may make sense as an actual determinant of an operator, while $z(\{q\})$ in the finite matrix, could always be zero in the function space considered, for example in $L^p$-spaces.
\end{remark}

\begin{remark}
	We stress that the determinant in Eq.~\eqref{rhs} is the determinant of a matrix indexed by a set of colors $Q$ and not by indices $i,j\in \{1,\ldots, n\}$. 
	Crucially, if $q_i = q_j$ for some $i \neq j$, the color $q= q_i = q_j$ gives rise to only \emph{one} row and column in the matrix. As a consequence, colors may repeat among the $q	_i$'s but the determinant still be non-zero. 
	
	This is consistent with the determinant in Eq.~\eqref{eq:finite-lagrange1} for finitely many variables. In analytic proofs \cite{good1960}, the determinant comes in via a complex change of variables in a contour integral. In particular, every variable $z_k$ appears only once in the determinant, even when we are interested in coefficients of monomials with powers $n_k\geq 2$. See also Eq.~\eqref{eq:smalldet0}.
\end{remark}

When $\Phi(z) =1$, we have $\Psi_n\equiv 0$ for all $n\geq 1$ and we obtain the following striking result. 

\begin{prop} \label{prop:magicformula}
For all $n\geq 1$ and $(q_1,\ldots,q_n)\in \mathbb X^n$, the term of order zero  in the formal power series 
\be  \label{eq:magic} 
	\frac{\delta^n}{\delta z(q_1)\cdots \delta z(q_n)} 
	\left\{\exp\left( \sum_{i=1}^n A(q_i;z) \right)
	 \det \left( \mathrm{Id} - \mathbb A_z \right)  
	\right\}
\ee
vanishes.
\end{prop} 

\noindent Note that the expression in the curly bracket also depends on $n$. Remark~\ref{rmk:nice} applies here as well.
These cancellations give rise to the inversion formula expressed in terms of trees as it will be reminded in Section~\ref{sec:treeinverse}.
Similar cancellations happen in \cite{abdesselam2003jacobian} 
where the formal inverse is expressed as a one-point correlation function $\rho^{(1)}(x)$ of an appropriately chosen
complex Bosonic gas.
In the latter case, cycles are cancelled by the partition function $Z$ leaving one tree with root color $x$.

Another special case of Theorem~\ref{thm:lagrange-good} corresponds to the choice $\Phi(z) = z(B)$ with $B\subset \mathbb X$ measurable (the analogue for finitely many variables is $\sum_{i\in B} z_i$).

\begin{cor} [Determinant formula for the inverse] \label{cor:inverse}
	For every $B\in \mathbb X$, the $n$-th coefficient of $\zeta(B;\nu)$, evaluated at $(q_1,\ldots, q_n) \in \mathbb X^n$, is equal to the $n$-th coefficient of 
	\be \label{eq:zbi}
		z(B) \exp\left( \sum_{i=1}^n A(q_i;z) \right)
		\det \left( \mathrm{Id} 
		- \mathbb A_z \right)  
	\ee
	evaluated at $(q_1,\ldots, q_n)$.
\end{cor} 

\noindent The proposition provides an analogue for~\eqref{eq:finite-lagrange2}.

\section{Lagrange-Good inversion. Proof of Theorem~\ref{thm:lagrange-good}} \label{sec:proofs}

\subsection{Enriched maps, cycle-rooted trees, weights}  
A key ingredient of the bijective proofs by Labelle and Gessel~\cite{labelle1981, gessel1987} is a graphical representation of mappings. 
Let us briefly recall some standard vocabulary, following \cite[Chapters~3.1 and~3.2]{bergeron-labelle-leroux1998book}. Let $V,S$ be disjoint finite set. 

A \emph{partial endofunction} on $U:=V\cup S$ with domain $V$ is a mapping $f:V \to U=V \cup S$.
Any partial endofunction is associated with a directed graph $G=(U,E)$, or \emph{digraph} for short, with vertex set $U$ and directed edges $E = \{ (v,f(v))\mid v\in U \}\subset U\times U$. The digraph $G$ may have self-edges $(v,v)$, called \emph{loops}. Loops in $G$ correspond to fixed points of the mapping, $f(v) =v$.
The graphs $G$ obtained in this way have two types of vertices: Vertices $v\in V$ have exactly one outgoing edge because $f$ maps $v$ to exactly one vertex $f(v)$. Vertices $v\in U\setminus V=S$ are \emph{sinks}, i.e., they have no outgoing edge, because they do not belong to the domain of $f$.

In both cases the out-degree in $G$ of every vertex is either $0$ or $1$. We call digraphs with this property \emph{functional digraphs}. Clearly there is a one-to-one correspondence between functional digraphs $G$ and partial endofunctions $f$ and we often identify them, cf.\ Figure~\ref{fig:endograph}.

 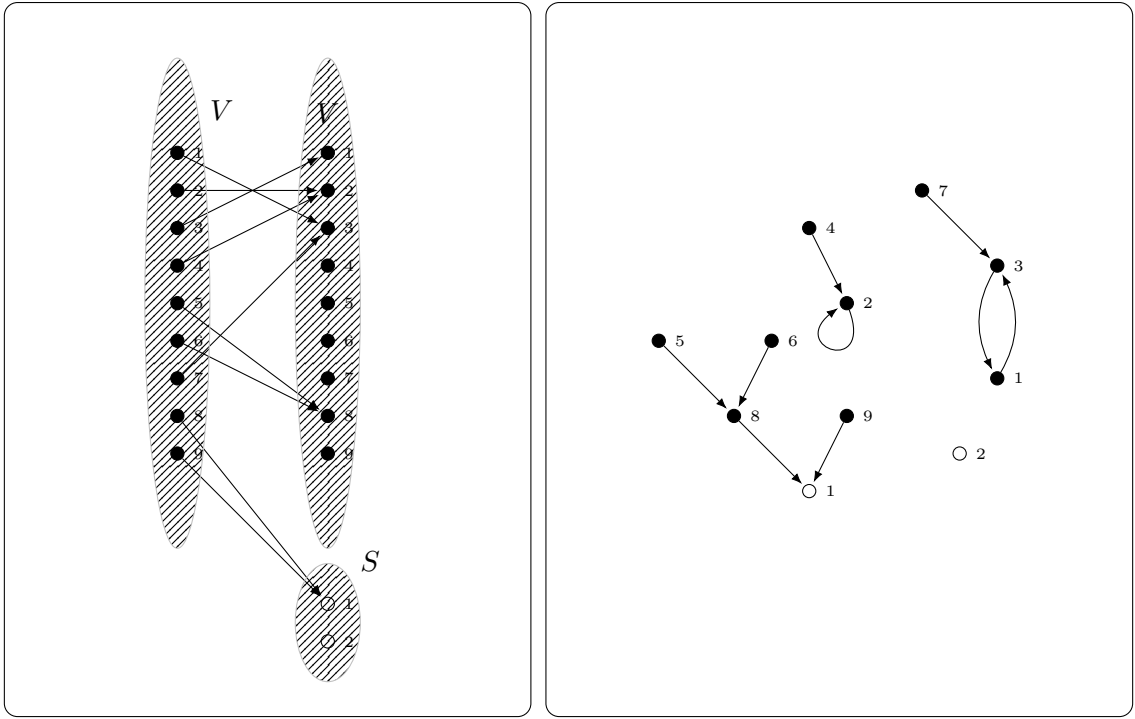
\begin{figure}
 \begin{tikzpicture}[scale=1,every label/.style ={font=\tiny}, 
   vertex/.style={minimum size=5pt,inner sep=0pt, fill,draw,circle},
   overtex/.style={minimum size=5pt,inner sep=0pt ,draw,circle},
open/.style={fill=none}, 
cross/.style={minimum size=2pt,fill,draw,circle},
sibling distance=1cm,
level distance=1.00cm, 
every fit/.style={ellipse,draw,inner sep=6.5pt}, 
arrow/.tip = {Latex[length=1.5mm,sep=1pt,bend]},
every fit/.style={ellipse,draw,inner sep=6pt}
]


\draw[rounded corners=5pt]  (-2,6) rectangle (5,-3.5);
\draw[rounded corners=5pt]  (5.2,6) rectangle (13,-3.5);

\begin{scope}[xshift=0.3cm]
\coordinate (dis) at (0,0.5);
\draw (0,0)  node (x9) [vertex,label={right:$9$}] {}
++ (dis) node (x8) [vertex,label={right:$8$}] {}
++ (dis) node (x7) [vertex,label={right:$7$}] {}
++ (dis) node (x6) [vertex,label={right:$6$}] {}
++ (dis) node (x5) [vertex,label={right:$5$}] {}
++ (dis) node (x4) [vertex,label={right:$4$}] {}
++ (dis) node (x3) [vertex,label={right:$3$}] {}
++ (dis) node (x2) [vertex,label={right:$2$}] {}
++ (dis) node (x1) [vertex,label={right:$1$}] {}
;

\draw (2,-2.5) node (y2) [overtex,label={right:$2$}] {}
++ (dis) node (y1) [overtex,label={right:$1$}] {};
\draw (2,0) node (xx9) [vertex,label={right:$9$}] {}
++ (dis) node (xx8) [vertex,label={right:$8$}] {}
++ (dis) node (xx7) [vertex,label={right:$7$}] {}
++ (dis) node (xx6) [vertex,label={right:$6$}] {}
++ (dis) node (xx5) [vertex,label={right:$5$}] {}
++ (dis) node (xx4) [vertex,label={right:$4$}] {}
++ (dis) node (xx3) [vertex,label={right:$3$}] {}
++ (dis) node (xx2) [vertex,label={right:$2$}] {}
++ (dis) node (xx1) [vertex,label={right:$1$}] {}
;
\draw (2,0) ++(dis)++(dis)++(dis) + (0.2,0) node (z1) {}
+ (-0.2,0) node (z2) {} ; 
\draw[-arrow] (x1) edge (xx3)
                      (x3) edge (xx1)
                      (x7) edge (xx3)
                      (x4) edge (xx2)
                      (x2) edge (xx2)
                      (x5) edge (xx8)
                      (x6) edge (xx8)
                      (x8) edge (y1)
                      (x9) edge (y1);
                      
  \node[color=lightgray, pattern=north east lines, fit= (y1) (y2), label={[font=\large]north east:$S$}] {}; 
  \node[color=lightgray, pattern=north east lines, fit = (x1) (x2) (x3) (x4) (x5) (x6) (x7) (x8) (x9), label={[font=\large]north east:$V$}] {};                   
   \node[color=lightgray, pattern=north east lines, fit = (xx1) (xx2) (xx3) (xx4) (xx5) (xx6) (xx7) (xx8) (xx9), label={[font=\large, label distance =-1cm]90:$V$}] {};                   
\end{scope}

\begin{scope}[xshift=8.7cm,yshift=-0.5cm]
\draw (0.5,1)  node (x9) [vertex,label={right:$9$}] {};
\draw (-1,1)  node (x8) [vertex,label={right:$8$}] {};
\draw (1.5,4)  node (x7) [vertex,label={right:$7$}] {};
\draw (-0.5,2)  node (x6) [vertex,label={right:$6$}] {};
\draw (-2,2)  node (x5) [vertex,label={right:$5$}] {};
\draw (0,3.5)  node (x4) [vertex,label={right:$4$}] {};
\draw (2.5,3)  node (x3) [vertex,label={right:$3$}] {};
\draw (0.5,2.5)  node (x2) [vertex,label={right:$2$}] {};
\draw (2.5,1.5)  node (x1) [vertex,label={right:$1$}] {};
\draw (0,0)  node (y1) [overtex,label={right:$1$}] {};
\draw (2,0.5)  node (y2) [overtex,label={right:$2$}] {};

\draw[-arrow] (x1) edge [bend right] (x3)
                      (x3) edge [bend right] (x1)
                      (x7) edge (x3)
                      (x4) edge (x2)
                      (x2) edge [out=-70, in=-150, distance=1cm] (x2)
                      (x5) edge (x8)
                      (x6) edge (x8)
                      (x8) edge (y1)
                      (x9) edge (y1);
\end{scope}

\end{tikzpicture}
\caption{\label{fig:endograph} Relation between partial endofunction  and the associated functional digraph. On the left hand side an endofunction in the usual representation of a functions is given. On the right hand side the functional digraph associated to the endofunction in the left box is drawn.}
\end{figure} 

In order to take into account the exponential $\exp( A(q;z))$ in our formulas, we enrich partial endofunctions with an additional structure: To each element $v\in V \cup S$  we add a set partition $ P_v$ of the elements of the preimage $f^{-1}(\{v\})$ (the preimage is often called \emph{fiber} of $v$).  This is a special case of the $R$-enriched structures often used in combinatorial proofs of  Lagrange inversion formula~\cite[Definitions~3.1.1 and~3.1.8]{bergeron-labelle-leroux1998book}. It is customary to represent $R$-structure as a structure on the incoming edges of the functional digraph; in our case, $ P_v$ is represented as a set partition of the incoming edges $(w,v)$. 
\begin{definition} \label{def:decomaps}
Let $V$ and $S$ be finite possibly empty disjoint sets.
 If $V$ is empty, we set $\mathcal{M}^S[V] = \mathcal{M}^S[\varnothing] := \varnothing$. If $V$ is non-empty, 
we define $\mathcal{M}^S[V]$ to be the collection of pairs $\bar f = (f, (P_v)_{v\in V\cup S})$ such that 
\begin{itemize} 
	\item $f$ is a map $f:V\to V\cup S$. 
	\item For each $v\in V\cup S$,  $P_v$ is a partition into non-empty sets of the preimage $f^{-1}(\{v\})$. If the preimage is empty we set $P_v = \varnothing$. 
\end{itemize} 
An element $\bar f= (f,( P_v)_{v\in V\cup S}) \in \mathcal M^S[V]$ is called \emph{enriched partial endofunction} on $V\cup S$, with domain $V$ and sink set $S$, or \emph{enriched map} for short. 
\end{definition}

\begin{figure}

\includegraphics{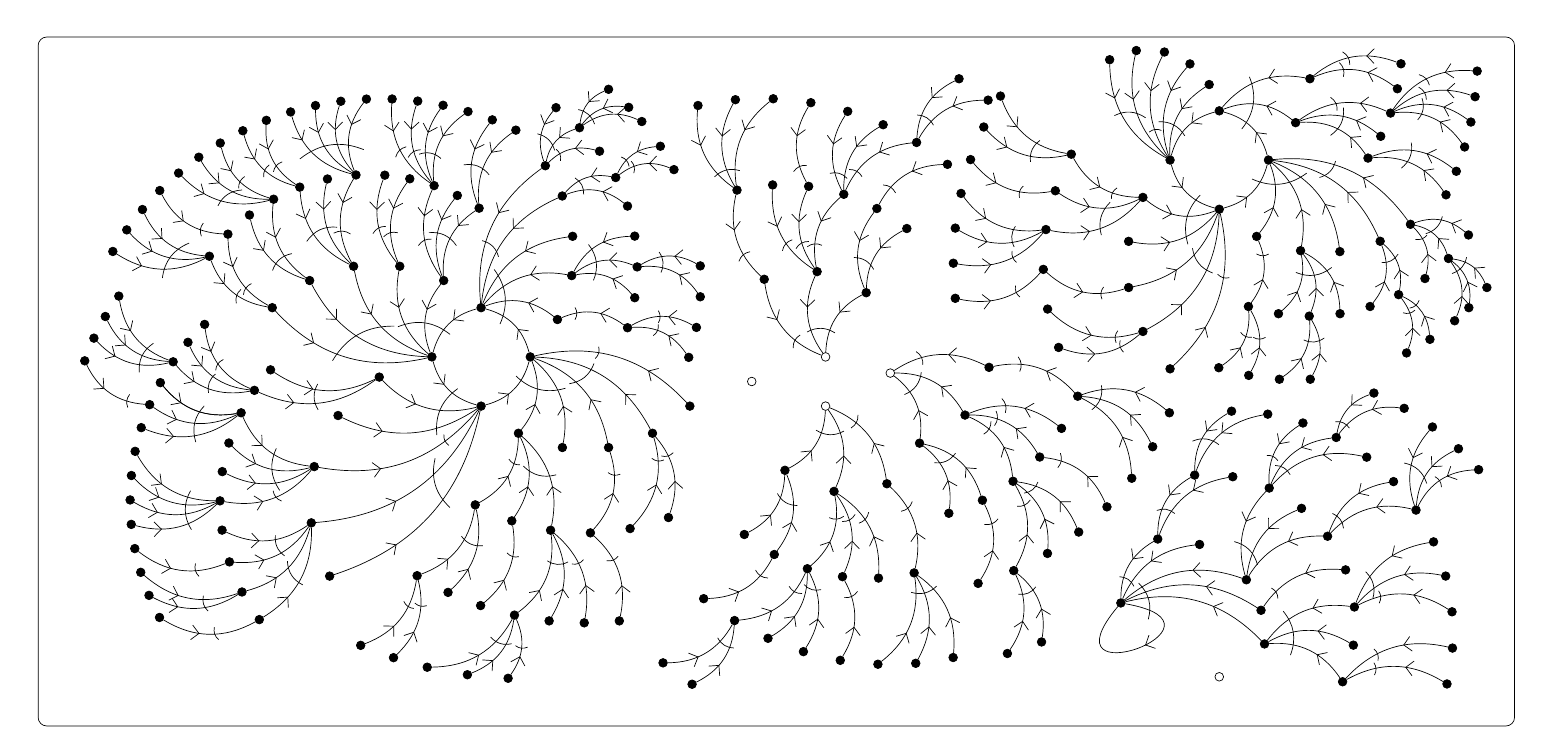}
\caption{\label{fig:ForestFull}Digraph of a typical enriched partial endofunction $\mathcal M^S[V]$  containing vertex-rooted trees and cycle-rooted trees. The vertices from $V$ are depicted as black filled circles, whereas the sinks $S$ are depicted as unfilled circles. The small lines crossing the edges of the graph represent the blocks $W$ of the partitions $P_v$. For more details cf. Figure~\ref{fig:weight}. 
}
\end{figure}

\noindent To lighten notation, if $S= \varnothing$ we drop the superscript $S$ and write $\mathcal M[V]$.  
If $S = \{\circ\}$ is a singleton we drop the braces in the notation and write $\mathcal M^\circ[V]$ instead of $\mathcal{M}^{\{\circ\}}[V]$.

A digraph $G$ is \emph{connected} if any two distinct vertices $v,w$ can be connected by a path from $v$ to $w$ or from $w$ to $v$. The connected components of  functional digraphs are of two types: trees and cycle-rooted trees, by which we mean the following. 

\begin{definition}
	Let $G= (V\cup S, E)$ be a functional digraph with domain $V$, that is a directed graph such that every vertex $v\in V$ has out-degree $1$ and every vertex $w\in S$ has out-degree zero. 
	\begin{enumerate}
		\item [(a)] A \emph{cycle} in $G$ is a sequence $v_1,\ldots, v_n$ in $V$ such that $(v_i,v_{i+1})\in E$ for all $i\in \{1,\ldots, n-1\}$ and $(v_n,v_1) \in E$. This includes loops ($n=1$). 
		\item [(b)] $G$ is a \emph{vertex-rooted tree} if it is connected and it has no cycle. Then $S$ is necessarily a singleton $ =\{\circ\}$, the root is the unique sink $\circ$.
		\item [(c)] $G$ is a \emph{cycle-rooted tree} if it is connected and it has  exactly one cycle. Then $S$ must be empty.The unique cycle in $G$ is called the \emph{root cycle} of $G$. 
	\end{enumerate} 
\end{definition} 

\noindent By some abuse of language an enriched map or graph are called vertex-rooted or cycle-rooted tree if the underlying graph $G$ is vertex- or cycle-rooted tree, respectively. The class of enriched maps $\bar f\in \mathcal M^\circ [V]$ that are vertex-rooted trees is denoted $\mathcal{T}^\circ(V)$.

Finally we define weights of enriched maps. For $\bar f = (f,(P_v)_{v\in V\cup S})\in \mathcal M^S[V]$ and $\vect x\in \mathbb X^{V\cup S}$, the weight of $\bar f$, given the coloring $\vect x$, is 
\be \label{mapweights}
	w( f, (P_v)_{v\in V\cup S};\vect x) := \prod_{\heap{v\in V\cup S:}{f^{-1}(v) \neq \varnothing}} \prod_{W\in P_v} A_{\#W}(x_v; (x_w)_{w\in W}).
\ee
We denote for short  $ {\vect x}_W:=(x_w)_{w\in W}$. For a graphical representation see Figure~\ref{fig:weight}.

\begin{figure}
\newcommand\arrowonline[1]{node 
[pos=#1,sloped,anchor=center,draw=none,fill=none,shape=rectangle,inner 
   sep=0pt,outer sep=0pt] {\tikz\draw[-{Straight Barb[length=1.5mm,sep=10pt,bend]},dash pattern=on 0pt off 2pt] 
(-1pt,0pt) -- (0pt,0pt);}} 
\newcommand\arrowonlinel[1]{node 
[pos=#1,sloped,anchor=center,draw=none,fill=none,shape=rectangle,inner 
   sep=0pt,outer sep=0pt] {\tikz\draw[{Straight Barb[length=1.5mm,sep=10pt,bend]}-,dash pattern=on 0pt off 2pt] 
(-1pt,0pt) -- (0pt,0pt);}}

\begin{tikzpicture}[scale=1,every label/.style ={font=\tiny}, 
   vertex/.style={minimum size=5pt,inner sep=0pt, fill,draw,circle},
   overtex/.style={minimum size=5pt,inner sep=0pt ,draw,circle},
open/.style={fill=none}, 
cross/.style={minimum size=2pt,fill,draw,circle},
sibling distance=1cm,
level distance=1.00cm, 
every fit/.style={ellipse,draw,inner sep=6.5pt}, 
arrow/.tip = {Latex[length=1.5mm,sep=10pt,bend]},
every fit/.style={ellipse,draw,inner sep=6pt}
]

\draw[rounded corners=5pt]  (-1,4) rectangle (7,-4);
\draw (0,0) node (z1) [overtex,label={north west:$x_0$}] {};
\path (70:3) arc [start angle=70, radius=3cm, end angle=-40] 
    node (x1) [pos=0.0,vertex,label={north:$x_1$}] {}
    node (x2) [pos=0.2,vertex,label={north east:$x_2$}] {}
    node (x3) [pos=0.4,vertex,label={right:$x_3$}] {}
    node (x4) [pos=0.6,vertex,label={south:$x_4$}] {}
    node (x5) [pos=0.8,vertex,label={south:$x_5$}] {}
    node (x6) [pos=1,vertex,label={south:$x_6$}] {};

\draw (x1) -- (z1) \arrowonlinel{0.3} ;
\draw (x2) -- (z1) \arrowonlinel{0.3} ;
\draw (x3) -- (z1) \arrowonlinel{0.3} ;
\draw (x4) -- (z1) \arrowonlinel{0.3} ;
\draw (x5) -- (z1) \arrowonlinel{0.3} ;
\draw (x6) -- (z1) \arrowonlinel{0.3} ;

\tikzmath{ real  \k{st} \k{en};  \k{st} =83; \k{en}= 12;\r{ad}=1;}
\draw[-,thick] (\k{st}:\r{ad}cm) node [label={$W_1$}]{} arc [start angle=\k{st} , end angle=\k{en}, radius=\r{ad}cm]; 

\tikzmath{ real  \k{st} \k{en};  \k{st} =13; \k{en}= -31; \r{ad}=1.2;}
\draw[-,thick] (\k{st}:\r{ad}cm)  arc [start angle=\k{st} , end angle=\k{en}, radius=\r{ad}cm]
node [pos=0.5,label={right:$W_2$}]{} ; 

\tikzmath{ real  \k{st} \k{en};  \k{st} =-31; \k{en}= -49;  \r{ad}=1;}
\draw[-,thick] (\k{st}:\r{ad}cm)  arc [start angle=\k{st} , end angle=\k{en}, radius=\r{ad}cm] node [label={$W_3$}]{};

\begin{scope}[shift=(x4)]
\draw (0,0) node (z1) 
{};
\path (50:2.5) arc [start angle=50, radius=2.5cm, end angle=-20] 
    node (x1) [pos=0.0,vertex,label={north:$x_7$}] {}
    node (x2) [pos=0.3,vertex,label={north east:$x_8$}] {}
    node (x3) [pos=0.6,vertex,label={right:$x_9$}] {}
    node (x4) [pos=0.9,vertex,label={right:$x_{10}$}] {};

\draw (x1) -- (z1) \arrowonlinel{0.3} ;
\draw (x2) -- (z1) \arrowonlinel{0.3} ;
\draw (x3) -- (z1) \arrowonlinel{0.3} ;
\draw (x4) -- (z1) \arrowonlinel{0.3} ;

\tikzmath{ real  \k{st} \k{en};  \k{st} =60; \k{en}= -2;\r{ad}=1;}
\draw[-,thick] (\k{st}:\r{ad}cm) node [label={[label distance=-0.3cm]south west:$W_4$}]{} arc [start angle=\k{st} , end angle=\k{en}, radius=\r{ad}cm];

\tikzmath{ real  \k{st} \k{en};  \k{st} =-5; \k{en}= -25;  \r{ad}=1.2;}
\draw[-,thick] (\k{st}:\r{ad}cm)  arc [start angle=\k{st} , end angle=\k{en}, radius=\r{ad}cm] node [label={[label distance=-0.1cm]left:$W_5$}]{}; 

\end{scope}

\begin{scope}[shift=(x5)]
\draw (0,0) node (z1) 
{};
\path (-5:2.5) arc [start angle=-5, radius=2.5cm, end angle=-40] 
    node (x1) [pos=0.0,vertex,label={right:$x_{11}$}] {}
    node (x2) [pos=0.5,vertex,label={right:$x_{12}$}] {}
    node (x4) [pos=1,vertex,label={right:$x_{13}$}] {};

\draw (x1) -- (z1) \arrowonlinel{0.3} ;
\draw (x2) -- (z1) \arrowonlinel{0.3} ;
\draw (x4) -- (z1) \arrowonlinel{0.3} ;

\tikzmath{ real  \k{st} \k{en};  \k{st} =8; \k{en}= -32;\r{ad}=1;}
\draw[-,thick] (\k{st}:\r{ad}cm) node [label={[label distance=-0.1cm]east:$W_6$}]{} arc [start angle=\k{st} , end angle=\k{en}, radius=\r{ad}cm];

\tikzmath{ real  \k{st} \k{en};  \k{st} =-30; \k{en}= -50;  \r{ad}=1.2;}
\draw[-,thick] (\k{st}:\r{ad}cm)  arc [start angle=\k{st} , end angle=\k{en}, radius=\r{ad}cm] node [label={[label distance=-0.2cm]south:$W_7$}]{}; 

\end{scope}


\begin{scope}[xshift=8.2cm]
\draw[rounded corners=5pt]  (-1,4) rectangle (7,-4);
\draw (0,0.3) node (z11) [overtex,label={north west:$$}] {};
\draw (0,0) node (z12) [overtex,label={west:$x_0$}] {};
\draw (0,-0.3) node (z13) [overtex,label={north west:$$}] {};
\path (70:3) arc [start angle=70, radius=3cm, end angle=-40] 
    node (x1) [pos=0.0,vertex,label={north:$x_1$}] {}
    node (x2) [pos=0.2,vertex,label={north east:$x_2$}] {}
    node (x3) [pos=0.4,vertex,label={right:$x_3$}] {}
    node (x4) [pos=0.6,vertex,label={south:$ $}] {}
    node (x5) [pos=0.8,vertex,label={south:$ $}] {}
    node (x6) [pos=1,vertex,label={south:$x_6$}] {};

\draw (x1) -- (z11) \arrowonlinel{0.3} ;
\draw (x2) -- (z11) \arrowonlinel{0.3} ;
\draw (x3) -- (z11) \arrowonlinel{0.3} ;
\draw (x4) -- (z12) \arrowonlinel{0.3} ;
\draw (x5) -- (z12) \arrowonlinel{0.3} ;
\draw (x6) -- (z13) \arrowonlinel{0.3} ;

 \node[fit= (z11) (z12) (z13), inner sep=1pt] {} ;

\tikzmath{ real  \k{st} \k{en};  \k{st} =85; \k{en}= 22;\r{ad}=1;}
\draw[-,thick] (\k{st}:\r{ad}cm) node [label={$W_1$}]{} arc [start angle=\k{st} , end angle=\k{en}, radius=\r{ad}cm]; 

\tikzmath{ real  \k{st} \k{en};  \k{st} =14; \k{en}= -33; \r{ad}=1.2;}
\draw[-,thick] (\k{st}:\r{ad}cm)  arc [start angle=\k{st} , end angle=\k{en}, radius=\r{ad}cm]
node [pos=0.5,label={right:$W_2$}]{} ; 

\tikzmath{ real  \k{st} \k{en};  \k{st} =-35; \k{en}= -60;  \r{ad}=1;}
\draw[-,thick] (\k{st}:\r{ad}cm)  arc [start angle=\k{st} , end angle=\k{en}, radius=\r{ad}cm] node [label={$W_3$}]{};

\begin{scope}[shift=($(x4)+(20:0.25)$)]
\draw (0,0.125) node (z11) [overtex,label={north west:$ $}] 
{};
\draw (0.04,-0.125) node (z12) [overtex,label={north west:$ $}] 
{};
\path (50:2.5) arc [start angle=50, radius=2.5cm, end angle=-20] 
    node (x1) [pos=0.0,vertex,label={north:$x_7$}] {}
    node (x2) [pos=0.3,vertex,label={north east:$x_8$}] {}
    node (x3) [pos=0.6,vertex,label={right:$x_9$}] {}
    node (x44) [pos=0.9,vertex,label={right:$x_{10}$}] {};

\draw (x1) -- (z11) \arrowonlinel{0.3} ;
\draw (x2) -- (z11) \arrowonlinel{0.3} ;
\draw (x3) -- (z11) \arrowonlinel{0.3} ;
\draw (x44) -- (z12) \arrowonlinel{0.3} ;

 \node[fit= (z11) (z12) (x4), inner sep=1pt] [label={south:$x_4$}] {} ;

\tikzmath{ real  \k{st} \k{en};  \k{st} =61; \k{en}= 3;\r{ad}=1;}
\draw[-,thick] (\k{st}:\r{ad}cm) node [label={[label distance=-0.3cm]south west:$W_4$}]{} arc [start angle=\k{st} , end angle=\k{en}, radius=\r{ad}cm];

\tikzmath{ real  \k{st} \k{en};  \k{st} =-7; \k{en}= -27;  \r{ad}=1.2;}
\draw[-,thick] (\k{st}:\r{ad}cm)  arc [start angle=\k{st} , end angle=\k{en}, radius=\r{ad}cm] node [label={[label distance=-0.1cm]left:$W_5$}]{}; 

\end{scope}

\begin{scope}[shift=($(x5)+(-30:0.2)$)]
\draw (0.125,0.125) node (z15) [overtex,label={north west:$ $}] 
{};
\draw (0,-0.125) node (z16) [overtex,label={north west:$ $}] 
{};

\path (-5:2.5) arc [start angle=-5, radius=2.5cm, end angle=-40] 
    node (x1) [pos=0.0,vertex,label={right:$x_{11}$}] {}
    node (x2) [pos=0.5,vertex,label={right:$x_{12}$}] {}
    node (x4) [pos=1,vertex,label={right:$x_{13}$}] {};

\draw (x1) -- (z15) \arrowonlinel{0.3} ;
\draw (x2) -- (z15) \arrowonlinel{0.3} ;
\draw (x4) -- (z16) \arrowonlinel{0.3} ;

 \node[fit= (z15) (z16) (x5), inner sep=1pt] [label={south:$x_5$}] {} ;

\tikzmath{ real  \k{st} \k{en};  \k{st} =8; \k{en}= -26;\r{ad}=1;}
\draw[-,thick] (\k{st}:\r{ad}cm) node [label={[label distance=-0.1cm]east:$W_6$}]{} arc [start angle=\k{st} , end angle=\k{en}, radius=\r{ad}cm];

\tikzmath{ real  \k{st} \k{en};  \k{st} =-34; \k{en}= -54;  \r{ad}=1.2;}
\draw[-,thick] (\k{st}:\r{ad}cm)  arc [start angle=\k{st} , end angle=\k{en}, radius=\r{ad}cm] node [label={[label distance=-0.2cm]south:$W_7$}]{}; 

\end{scope}

\end{scope}


%
%
%
%
%

\end{tikzpicture}

\caption{\label{fig:weight}
The weight~\eqref{mapweights} of an enriched digraph is a product of contributions of building blocks. The left panel represents an enriched digraph with weight
\[
	A_{3}(x_0 ; x_1, x_2 ,x_3) A_{2}(x_0 ; x_4, x_5)
A_{1}(x_0 ; x_6)A_{2}(x_4 ; x_7, x_8) A_{1}(x_4 ; x_{10})
A_{2}(x_5 ; x_{11}, x_{12}) A_{1}(x_5 ; x_{13}).
\]
The right panel separates the graph into the single building blocks, each of which corresponds to a factor  $A_{\#W}(x_v ; (x_w)_{w \in W})$ in the graph weight. The partitions are given by $W_1 = \{ 1,2,3\}$, $W_2=\{ 4, 5 \}$, $W_3 =\{ 6 \}$, $W_4=\{ 7,8,9\}$, $W_5 = \{ 10\}$, $W_6=\{ 11, 12\}$, $W_7 = \{13\}$ and the partitions are $P_0 = \{ W_1 ,W_2 ,W_3\}$, $P_4 = \{ W_4, W_5 \}$, and $P_5 = \{ W_6, W_7 \}$. 
}
\end{figure}
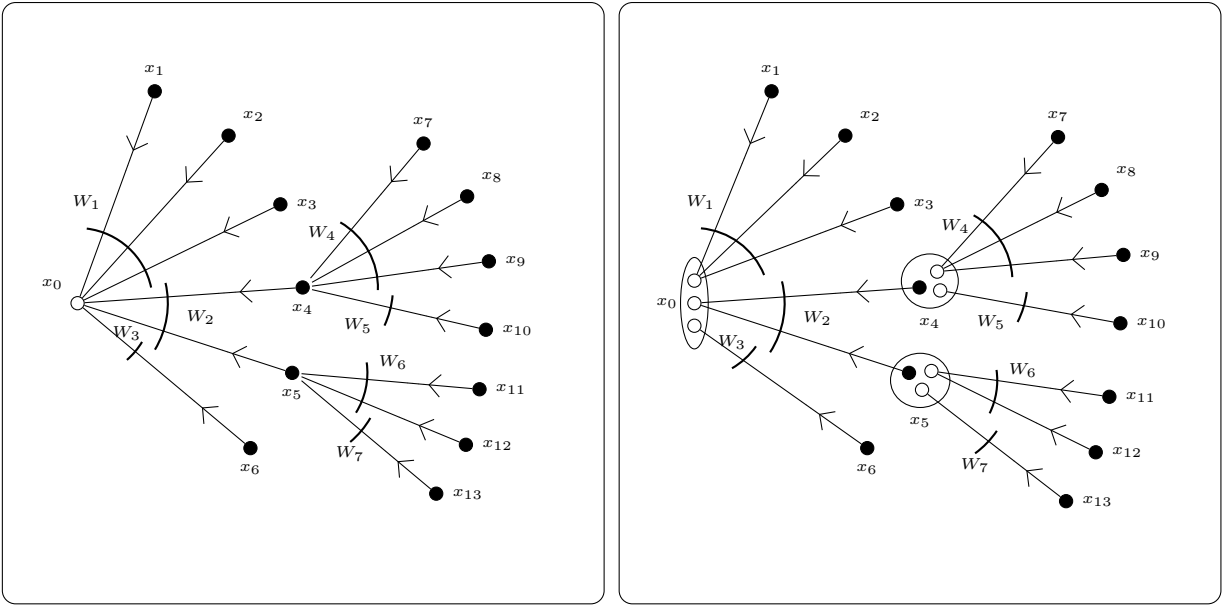

\subsection{The tree formula for the inverse power series}  \label{sec:treeinverse}
 
For our further calculation we need a combinatorial expression for the inverse power series. Such a representation can be obtained directly without using Lagrange-Good type formulas. We briefly recall the tree formula for the inverse power series proven in~\cite{jansen-kuna-tsagkaro2019virialinversion}. Let 
\begin{equation}\label{tnn}
	t_n(x_\circ; x_1,\ldots,x_n):=\sum_{\bar f\in \mathcal{T}^\circ[n]}
	 w\left(\bar f; (x_j)_{j\in [n]\cup\{\circ\}}\right)
\end{equation}
and 
\be \label{treegen}
	T_q^\circ(\nu) := 1+\sum_{n=1}^\infty \frac{1}{n!} \int_{\mathbb X^n} t_n(x_\circ; x_1,\ldots,x_n) \nu (\dd x_1)\cdots \nu(\dd x_n)\quad x_\circ =q.
\ee 

\begin{lemma} \label{lem:treefip}
	The family $(T_q^\circ(\nu))_{q\in \mathbb X}$ fulfils the following  functional equation as formal power series
\be \label{tree-eq-again}
	T_q^\circ(\nu) = \exp\left( \sum_{n=1}^\infty \frac{1}{n!} \int_{\mathbb X^n} A_n(q;x_1,\ldots,x_n) T_{x_1}^\circ(\nu)\cdots T_{x_n}^\circ( \nu)  \nu(\dd x_1)\cdots  \nu(\dd x_n)\right).
\ee
\end{lemma} 

\noindent For finite spaces $\mathbb X$, the lemma follows from Theorem~2 in Chapter~3.2 on implicit species in the book by Bergeron, Labelle, and Leroux \cite{bergeron-labelle-leroux1998book}. 

\begin{proof}
	The lemma follows from~\cite[Lemma~2.1 and Proposition~2.6]{jansen-kuna-tsagkaro2019virialinversion}, we give a self-contained proof for the reader's convenience.
 Calling $\{V_1, \ldots , V_m\}$ the partition $P_\circ$ of the vertices incident to the root $\circ$, we can rewrite $t_n$ using the tree structure of the associated graph as
\begin{equation}
\sum_{\bar f\in \mathcal{T}^\circ[n]} w\left(\bar f; (x_j)_{j\in [n]\cup\{\circ\}}\right) = \sum_{m=1}^n  \sum_{\substack{\{V_1,\ldots,V_m\}
\\ \text{partition of}\,[n]}} \prod_{i=1}^m
B(q;V_i), \label{eqtreedec}
\end{equation}
where we defined:
\begin{equation}\label{insidedd}
B(q;V_i):=
\sum_{\substack{L_i\subset V_i \\ L_i\neq \varnothing}}
A_{\# L_i}(q;\vect x_{L_i})\sum_{\substack{(J_k)_{k\in L_i} \\ \text{partition of}\, V_i\setminus L_i}} \prod_{k\in L_i} t_{\# J_k}(x_k;\vect x_{J_k}).
\end{equation}
(Compare with the calculations (formulas (26)-(28)) 
in the proof of Theorem 1 in \cite{abdesselam2003jacobian}].)
Note that according to \eqref{eq:formal-product} the expression $B(q;V_i)$ is the $\# V_i$-th. coefficient  in powers of $\nu$ of 
\begin{equation}\label{eq:expef}
 \sum_{n=1}^\infty \frac{1}{n!} \int_{C^n} A_n(q;x_1,\ldots,x_n) T_{x_1}^\circ(\nu)\cdots T_{x_n}^\circ( \nu) \nu(\dd x_1)\cdots \nu(\dd x_n) .
 \end{equation}
According to~\eqref{eq:formal-exponential}, the right hand side of \eqref{eqtreedec} is the $n$-th coefficient of the exponential of  \eqref{eq:expef}.
\end{proof} 

When dealing with generating functions of rooted trees we have two choices, either look at trees rooted in a ghost or trees rooted in a labelled vertex. In the univariate case (no colors) these give rise to two different generating functions $T^\circ(\nu)$ and $T^\bullet (\nu)$ related by $T^\bullet (\nu)= \nu T^\circ (\nu)$. The multivariate case (finitely many colors) gives rise to a relation of the type $ T_i^\bullet (\nu)= \nu_i T^\circ_i(\vect \nu) $, with $i$ the color of the root. This last relation should remind the reader of the relation $ \zeta(\dd q;\nu) = \nu(\dd q) T_q^\circ (\nu) $ used to define $\zeta(\dd q;\nu)$ in~\eqref{eq:zetadef}. As a consequence, we should expect that the measure valued series $\zeta(\dd q;\nu)$ corresponds to rooted trees with the root integrated over. 

The next lemma makes this statement precise. It says that the series $\zeta(B;\nu)$ is given by a sum over rooted trees with root color in $B$. 

\begin{lemma} \label{lem:treeformula}
	Let $\zeta(\dd q;z) = T_q^\circ(q;z) z(\dd q)$. 
	The $n$-th coefficient of $\zeta(B;\nu)$, evaluated at $(q_1,\ldots, q_n)$, is equal to 
	\be \label{eq:treeformula}
		\zeta_n(B;q_1,\ldots, q_n) = \sum_{r=1}^n \1_B(q_r) t_{n-1}\bigl( q_r;(q_i)_{i\in [n]\setminus\{r\}}\bigr).
	\ee
	and the term of order zero vanishes, $\zeta_0(B) =0$. 
\end{lemma} 

\begin{proof}
	The lemma is an immediate consequence of the definition~\eqref{eq:zetadef} of $\zeta(\dd q;\nu)$ and Eq.~\eqref{eq:formal-radon}. 
\end{proof}

\begin{theorem}\label{thm:inversetree}
	For $\rho$ given in \eqref{rho} we have $(\zeta \circ \rho)(\dd q;z) = z(\dd q)$ and $ (\rho\circ \zeta)(\dd q;\nu) = \nu(\dd q)$ as an equality of measure-valued formal power series, with the composition defined in~\eqref{eq:fmc}. 
\end{theorem} 

\begin{proof}[Proof sketch]
	By the definition of $\rho(\dd q; z) = \exp(-A(q;z)) z(\dd q)$ and Lemma~\ref{lem:treefip}, we have 
	\be
			\bigl(\rho \circ \zeta\bigr)(\dd q;\nu) = \e^{- A(q;\zeta[\nu])} \zeta(\dd q;\nu) = \e^{- A(q;\zeta[\nu])}  T_q^\circ (\nu) \nu(\dd q)  = \nu(\dd q).
	\ee
   The previous chain of equalities is formal but it can be rigorously justified by properties of operations on formal power series (e.g. associativity of the product), defined in terms of coefficients only; we leave the details to the reader.    
   The  equality $(\zeta \circ \rho)(\dd q;z) = z(\dd q)$ is proven by a similar argument and another fixed point equation, see~\cite[Lemma~2.12]{jansen-kuna-tsagkaro2019virialinversion}.
\end{proof}

\subsection{Cycle-rooted forests. Proof of Proposition~\ref{prop:magicformula}}

For the proof of Proposition~\ref{prop:magicformula}, we relate the right-hand side of  Eq.~\eqref{eq:magic} to a sum over cycle-rooted forests and look for combinatorial cancellations. We start with the exponential. Let $E_k^n(\vect{q};\vect{x})$ be the coefficients of the exponential  appearing in Proposition~\ref{prop:magicformula}, i.e., 
\begin{equation} \label{ekdef}
	\exp\left(\sum_{i=1}^n A(q_i;z)\right) 
	 =  1+ \sum_{k=1}^\infty \frac{1}{k!} \int_{\mathbb X^k} E^n_k(q_1,\ldots,q_n;x_1,\ldots,x_k)  z(\dd x_1)\cdots  z(\dd x_k).
\end{equation} 

\noindent

Figure~\ref{fig:ForestFull} depicts a forest of vertex-rooted and cycle-rooted trees. The vertices $[n]\setminus I$ are represented as white.

\begin{lemma} \label{lem:ek}
Let $E^n_k$ be as in~\eqref{ekdef}. 
For $n\in \N$, $\vect{q}\in \mathbb X^n$, and $I\subset[n]$, we have 
\be\label{E}
	E^n_{\#I} \bigl(q_1,\ldots, q_n;(q_i)_{i\in I}\bigr) = \sum_{\bar f\in  \mathcal{M}^{[n]\setminus I}(I)} w(\bar f;q_1,\ldots, q_n).
\ee
\end{lemma} 

\noindent Notice that Lemma~\ref{lem:ek} does \emph{not} address general coefficients $E^n_{\#I}(q_1,\ldots, q_n;(x_i)_{i\in I})$, but only the special case $x_i = q_i$. Roughly, $E_{\#I}^n(q_1,\ldots, q_n; (x_i)_{i\in I})$ is a sum over enriched maps from $I$ to $\{1,\ldots, n\}$, but if $x_i \neq q_i$ there are color-conflicts and the color-dependent weight of a map is inappropriate to give a representation.

\begin{proof} 
By $\exp\left(\sum_{i=1}^n A(q_i;z)\right)  = \prod_{i=1}^n \exp \left( A(q_i;z)\right) $ and ~\eqref{eq:formal-higher-product} we can write for all $\vect{x}\in \mathbb X^I$
\be
	E^n_{\#I} (q_1,\ldots,q_n; \vect x_I) = \sum_{I_1,\ldots,I_n} \prod_{\substack{1\leq k \leq n\\ I_k\neq \varnothing}} E^1_{\# I_k}(q_k;\vect x_{I_k}), 
	\ee	
where $I_1,\ldots,I_n$ are pairwise disjoint subsets, with $I_k = \varnothing$ allowed, whose union is $I$ and ${\vect x}_I = (x_i)_{i \in I}$. By ~\eqref{composition2} we have that:
\be 
E^1_{\# I_k}(q_k;\vect x_{I_k}) =\sum_{P_k \in \mathcal{P}(I_k) }\prod_{W\in P_k} A_{\# W}(q_k;\vect x_W),
\ee
where $\mathcal{P}(I_k)$ is the collection of  set partitions $\{W_1,\ldots,W_r\}$ of $I_k$ into non-empty sets $W_i$. Combining the two we obtain that
\be\label{ecoeff}
	E^n_{\#I} (q_1,\ldots,q_n; \vect x_I)  =\sum_{\substack{I_1,\ldots,I_n\\ P_1,\ldots, P_n}} \prod_{\substack{1\leq k \leq n\\ I_k\neq \varnothing}} \prod_{W\in P_k} A_{\# W}(q_k;\vect x_W).
\ee
Each $n$-tuple $(I_1,\ldots,I_n)$ gives rise to a map $f:I\to [n]$ by defining $f^{-1}(\{k\}) = I_k$; the correspondence is clearly one-to-one.
Considering $\bar f = \left( f , (P_k)_{k \in [n]} \right)$ and $\vect x_I=\vect q_I$ we recognize the weight $w(\bar f,\vect q)$
and obtain \eqref{E}.
\end{proof} 

\noindent Next we turn to the determinant on the right-hand side of  Eq.~\eqref{eq:magic}.  Write 
\be \label{dkdef}
	\det \left( \mathrm{Id} 
		- z(\dd q)\frac{\delta}{\delta z(q')} A(q;z) \right)  = 1 + \sum_{k=1}^\infty \frac{1}{k!}\int_{C^k} D_k(x_1,\ldots,x_k)  z^n(\dd \vect x) ,
\ee
where $ Q= \{q_i\mid i=1,\ldots, n\}$ is defined as in \eqref{Q}. 
We start with the interpretation of the left-hand side as a formal power series, cf. \eqref{eq:fredformal},
\begin{align}
& \det \left( \mathrm{Id} 
		- z(\dd q)\frac{\delta}{\delta z(q')} A(q;z) \right) \\
		&   = 1 + \sum_{n=1}^\infty \frac{(-1)^n}{n!}\int_{\mathbb X^n}  \det\left( \left( \frac{\delta}{\delta z(p_i)} A(p_j;z) \right)_{i,j=1, \ldots ,n} \right) z(\dd p_1) \ldots z(\dd p_n) \nonumber
\end{align}
and next we use 
the combinatorial interpretation of the matrix element
\begin{equation}\label{M}
	M_{p,p'}(z) = \frac{\delta}{\delta z(p')} A(p;z),
\end{equation}
which by  Eq.~\eqref{eq:aqdev} takes the form
\be
	M_{p,p'}(z)  = \sum_{k=0}^\infty \frac{1}{k!}\int_{\mathbb X^k}
	 		A_{k}(p;p',\vect x_{[k]})
			  z(\dd x_1)\cdots  z(\dd x_k).		 	 
\ee
%
In $M_{p,p'}(z)$ we recognize the generating function for digraphs $G$ with vertices $\circ, \circ'$, $1,\ldots, k$ that have edges $(\circ, \circ')$ and $(\circ, j)$, $j = 1, \ldots k$. The vertex $\circ$ has color $p$, the vertex $\circ'$ color $p'$, and $j$ the color $x_j$. 

We call such a digraph with the associated coloring a \emph{spike of type $(p,p')$} and the edge $(\circ, \circ')$ the \emph{base edge} of the spike.  
A spike with base edge $(\circ, \circ')$ is identified with the enriched graph $\bar G$ that has the trivial partition consisting of a single block, $ P_{\circ} = \{\circ', \vect x_{[k]} \} $. 
 See the left panel of Figure~\ref{fig:crown} for an example of a spike.
 
Concatenating base edges of spikes in a circular fashion gives rise to an object consisting of a base cycle and additional incoming edges.

\begin{definition} \label{def:crown}
Let $V$ be a finite non-empty set and $\bar f = (f,(P_v)_{v\in V}) \in \mathcal{M}[V]$. We call $\bar f$ a \emph{crown} if there exists a finite set $B\subset V$ (the \emph{base} of the crown)
such that: 
\begin{itemize}
	\item $f(V) = B$,
	\item $f$ restricted to $B$ is a cycle,
	\item for all $r\in B$, $P_r = \{f^{-1}(\{r\})\}$, that is, the partition $P_r$  consists of a single block.
\end{itemize}
Let $\mathcal{K}[V]$ be the collection of crowns on $V$.
\end{definition} 

The right panel in Figure~\ref{fig:crown} shows are crown. The simplest crown is a loop.

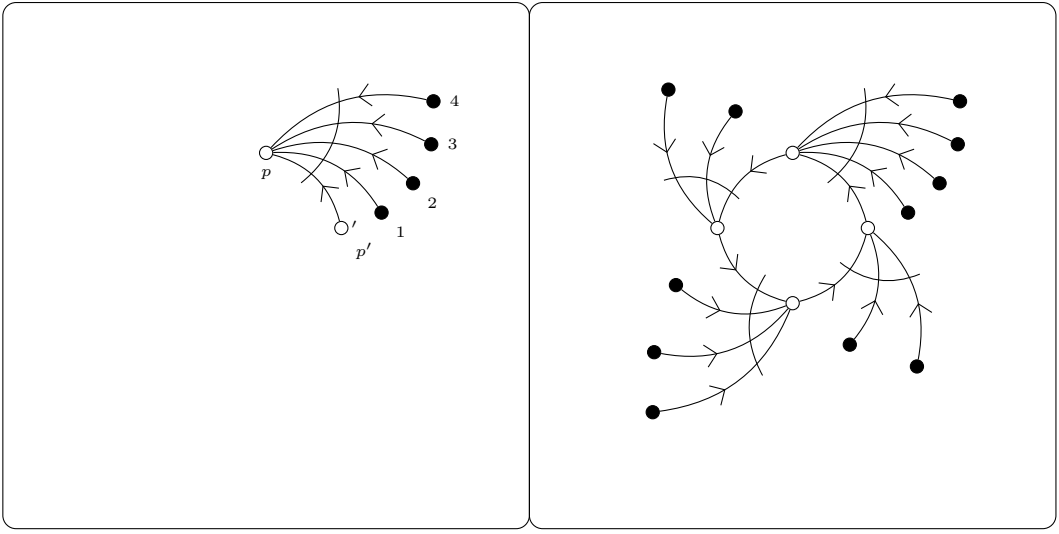
\begin{figure}

\newcommand\arrowonline[1]{node 
[pos=#1,sloped,anchor=center,draw=none,fill=none,shape=rectangle,inner 
   sep=0pt,outer sep=0pt] {\tikz\draw[-{Straight Barb[length=1.5mm,sep=10pt,bend]},dash pattern=on 0pt off 2pt] 
(-1pt,0pt) -- (0pt,0pt);}} 
\newcommand\arrowonlinel[1]{node 
[pos=#1,sloped,anchor=center,draw=none,fill=none,shape=rectangle,inner 
   sep=0pt,outer sep=0pt] {\tikz\draw[{Straight Barb[length=1.5mm,sep=10pt,bend]}-,dash pattern=on 0pt off 2pt] 
(-1pt,0pt) -- (0pt,0pt);}} 

\newcommand\markpar{{markings,
    mark=at position 0.5 with {{Straight Barb[length=1.5mm]}}}}

\begin{tikzpicture}[scale=1,every label/.style ={font=\tiny}, 
   vertex/.style={minimum size=5pt,inner sep=0pt, fill,draw,circle,},
   mylabel/.style={label={[transparent]#1}},
   overtex/.style={minimum size=5pt,inner sep=0pt ,draw,circle},
open/.style={fill=none}, 
cross/.style={minimum size=2pt,fill,draw,circle},
sibling distance=1cm,
level distance=1.00cm, 
every fit/.style={ellipse,draw,inner sep=6.5pt}, 
arrow/.tip = {Latex[length=1.5mm,sep=10pt,bend]},
every fit/.style={ellipse,draw,inner sep=6pt},
glines/.style={%
   postaction={decorate,decoration={
       markings, mark=at position #1 with {\arrow{Straight Barb[length=1.5mm]}}
       }
       }
}       
]

\draw[rounded corners=5pt]  (-1,3) rectangle (6,-4);
\draw[rounded corners=5pt]  (6,3) rectangle (13,-4);

\begin{scope}[xshift=2.5cm,rotate=180]
\coordinate (z1) at (0,0);
\tikzmath{\r{ad}=1;}
 
\path 
(180:\r{ad}cm) node (x3)  [overtex,label={south east:$p'$}, label={east:$\! \! '$}] {}
(270:\r{ad}cm) node (x4)  [overtex,label={south:$p$}] {};

\path[name path=x3x4] (x3) to [bend right] (x4);
\draw[glines=0.4] (x3) to [bend right] (x4); 

\path ($(x3)$) arc [start angle=100, radius=1.5cm, end angle=210] 
    node (x32) [pos=0.2,vertex,label={south east:$1$}] {}
    node (x33) [pos=0.4,vertex,label={south east:$2$}] {}
    node (x34) [pos=0.6,vertex,label={east:$3$}] {}
    node (x35) [pos=0.8,vertex,label={east:$4$}] {}
   ;
    
        \path[name path=transs] 
           ($(x4)+0.4*($(x3) - (x4)$)$) arc [start angle=140, radius=2.1cm, end angle=210];

     \path[name path=x32x4] (x32) to [bend right]  (x4);
      \draw[use path=x32x4, glines=0.4] ;
           \path[name path=x33x4] (x33) to [bend right]  (x4);
      \draw[use path=x33x4, glines=0.3] ;

    \path[name path=x35x4] (x35) to [bend right]  (x4);
      \draw[use path=x35x4, glines=0.4] ;
  
     \path[name path=x34x4] (x34) to [bend right]  (x4);
      \draw[use path=x34x4, glines=0.35] ;
       
       \draw[name intersections={of=x3x4 and transs, name=i, total=\t}, name intersections={of=x35x4 and transs, name=j, total=\tr}] 
     ($(i-\t) - 0.2*($(j-\tr)- (i-\t)$)$) to [bend right] ($(j-\tr) - 0.25*($(i-\t)- (j-\tr)$)$);
  
\end{scope}


\begin{scope}[xshift=9.5cm,yshift=0cm,rotate=180]
\coordinate (z1) at (0,0);
\tikzmath{\r{ad}=1;}
 
\path (0:\r{ad}cm) node (x1)  [overtex, mylabel={north:$x_1$}] {}
(90:\r{ad}cm) node (x2)  [overtex, mylabel={north:$x_2$}] {}
(180:\r{ad}cm) node (x3)  [overtex,mylabel={north:$x_3$}] {}
(270:\r{ad}cm) node (x4)  [overtex,mylabel={north:$x_4$}] {};

\path[name path=x3x4] (x3) to [bend right] (x4);
\draw[glines=0.4] (x3) to [bend right] (x4); 
\path[name path=x2x3] (x2) to [bend right] (x3);
      \draw[use path=x2x3, glines=0.4] ;
      \path[name path=x4x1] (x4) to [bend right] (x1);
      \draw[use path=x4x1, glines=0.4] ;
      \path[name path=x1x2] (x1) to [bend right] (x2);
      \draw[use path=x1x2, glines=0.4] ;
\path ($(x2)$) arc [start angle=45, radius=3cm, end angle=135] 
    node (x22) [pos=0.2,vertex,mylabel={north east:$22$}] {}
    node (x23) [pos=0.4,vertex,mylabel={right:$23$}] {}
    ;
     \path[name path=x22x3] (x22) to [bend right] (x3);
      \draw[use path=x22x3, glines=0.4] ;
     \path[name path=x23x3] (x23) to [bend right] (x3);
      \draw[use path=x23x3, glines=0.4] ;
%
     \path[name path=transs] ($(x3)+0.4*($(x2) - (x3)$)$) arc [start angle=65, radius=2cm, end angle=155];
     \draw[name intersections={of=x2x3 and transs, name=i, total=\t}, name intersections={of=x23x3 and transs, name=j, total=\tr}] 
     ($(i-\t) - 0.25*($(j-\tr)- (i-\t)$)$) to [bend right] ($(j-\tr) - 0.25 *($(i-\t)- (j-\tr)$)$);
  \path[name path=transs] ($(x3)+0.4*($(x2) - (x3)$)$) arc [start angle=65, radius=2cm, end angle=155];
      
%
%

\path ($(x3)$) arc [start angle=100, radius=1.5cm, end angle=210] 
    node (x32) [pos=0.2,vertex,mylabel={north east:$32$}] {}
    node (x33) [pos=0.4,vertex,mylabel={right:$33$}] {}
    node (x34) [pos=0.6,vertex,mylabel={south:$34$}] {}
    node (x35) [pos=0.8,vertex,mylabel={south:$35$}] {}
   ;
    
        \path[name path=transs] 
           ($(x4)+0.4*($(x3) - (x4)$)$) arc [start angle=140, radius=2.1cm, end angle=210];

     \path[name path=x32x4] (x32) to [bend right]  (x4);
      \draw[use path=x32x4, glines=0.4] ;
           \path[name path=x33x4] (x33) to [bend right]  (x4);
      \draw[use path=x33x4, glines=0.3] ;

    \path[name path=x35x4] (x35) to [bend right]  (x4);
      \draw[use path=x35x4, glines=0.4] ;
  
     \path[name path=x34x4] (x34) to [bend right]  (x4);
      \draw[use path=x34x4, glines=0.35] ;
       
       \draw[name intersections={of=x3x4 and transs, name=i, total=\t}, name intersections={of=x35x4 and transs, name=j, total=\tr}] 
     ($(i-\t) - 0.2*($(j-\tr)- (i-\t)$)$) to [bend right] ($(j-\tr) - 0.25*($(i-\t)- (j-\tr)$)$);
%
%

\path ($(x4)$) arc [start angle=-135, radius=3cm, end angle=-45] 
    node (x42) [pos=0.2,vertex,mylabel={north east:$42$}] {}
    node (x43) [pos=0.4,vertex,mylabel={right:$43$}] {}
   ;
    
    \path[name path=x1x42] (x42) to [bend right] (x1);
      \draw[use path=x1x42, glines=0.4] ;
\path[name path=x1x43] (x43) to [bend right] (x1);
      \draw[use path=x1x43, glines=0.4] ;

  \path[name path=transs] 
           ($(x1)+0.3*($(x4) - (x1)$)$) arc [start angle=-135, radius=1.0cm, end angle=-2];
      \draw[name intersections={of=x4x1 and transs, name=i, total=\t}, name intersections={of=x1x43 and transs, name=j, total=\tr}] 
 ($(i-\t) - 0.2*($(j-\tr)- (i-\t)$)$) to [bend right] ($(j-\tr) - 0.3*($(i-\t)- (j-\tr)$)$);

%
%
\path ($(x1)$) arc [start angle=-45, radius=3cm, end angle=45] 
    node (x12) [pos=0.2,vertex,mylabel={north east:$12$}] {}
    node (x13) [pos=0.4,vertex,mylabel={right:$13$}] {}
    node (x14) [pos=0.57,vertex,mylabel={south:$14$}] {}
     ;

\path[name path=transs] ($(x2)+0.3*($(x1) - (x2)$)$) arc [start angle=-25, radius=1.5cm, end angle=55];

\path[name path=x12x2] (x12) to [bend right] (x2);
\draw[use path=x12x2, glines=0.4] ;
\path[name path=x13x2] (x13) to [bend right] (x2);
\draw[use path=x13x2, glines=0.4] ;
\path[name path=x14x2] (x14) to [bend right] (x2);
\draw[use path=x14x2, glines=0.4] ;


\draw[name intersections={of=x1x2 and transs, name=i, total=\t}, name intersections={of=x14x2 and transs, name=j, total=\tr}] 
 ($(i-\t) - 0.3*($(j-\tr)- (i-\t)$)$) to [bend right] ($(j-\tr) - 0.3*($(i-\t)- (j-\tr)$)$);

\end{scope}

\end{tikzpicture}

\caption{\label{fig:crown}The left panel shows a spike  with $k=4$. The right panel shows a crown, where the set of white vertices the base $B$ and $V$ (as in Definition~\ref{def:crown}) is the set of all white and black vertices. The partitions $P_r$ are symbolized by the small curved lines crossing the edges. }
\end{figure} 

\noindent The next lemma expresses the relevant coefficient of the determinant as a sum over assemblies of crowns, with an additional factor $(-1)^s$ where $s$ is the number of crowns. The factor $(-1)^s$ is a crucial ingredient to cancellations in the proof of Proposition~\ref{prop:magicformula}.

\begin{lemma}\label{lem:dk}
Let $D_k$ be as in~\eqref{dkdef}. 
Then 
	\be\label{Ds}
	D_{k}(\vect{q}_{[k]}) =\sum_{s=1}^{k} (-1)^s \sum_{ \{V_1,\ldots, V_s\}} \prod_{i=1}^s \left(\sum_{\bar f_i \in \mathcal{K}[V_i]} w(\bar f; \vect{q}_{V_i}) \right)
\ee
where the sum is over set partitions $\{V_1,\ldots,V_s\}$ of $[k]$.
\end{lemma}

\begin{proof} 
Using the definition as formal power series \eqref{eq:fredformal}, we see that the $m$-th coefficient of the determinant is 
\begin{align}\label{eq:detTK}
& \det \left( \mathrm{Id} 
		- z(\dd q)\frac{\delta}{\delta z(q')} A(q;z) \right)_m\hspace{-0.3cm}(p_1, \ldots , p_m) = (-1)^m  \det\left( \left( \frac{\delta}{\delta z(p_i)} A(p_j;z) \right)_{i,j=1, \ldots ,m} \right) .
\end{align}
The definition of the determinant gives  
\begin{equation}\label{insidedet}
	\det(-M_{p_i p_j})_{1\leq i,j\leq m} = \sum_{\sigma \in \mathfrak{S}_m} \mathrm{sgn}\,(\sigma) \prod_{i=1}^m (- M_{p_{\sigma(i)},p_{i}}),
\end{equation} 
which we may also write as 
\be \label{insidedet2}
		\det(-M_{p_i p_j})_{1\leq i,j\leq m} = \sum_{\sigma \in \mathfrak{S}_m} (-1)^{s(\sigma)} \prod_{i=1}^m  M_{p_{\sigma(i)},p_{i}}
\ee
with $s(\sigma)$ the number of cycles in the cycle decomposition of the permutation $\sigma$. 
Indeed, let  $\sigma=\sigma_1\cdots\sigma_s$ be the cycle decomposition into $s$ cycles
of respective lengths $\ell(\sigma_i)$. Then, $\sum_{i=1}^s \ell(\sigma_i)=m$ and 
\be\label{cs}
	(-1)^m\mathrm{sgn}(\sigma)=(-1)^m\prod_{i=1}^s(-1)^{\ell(\sigma_i)-1}=(-1)^s. 
\ee
The sum in~\eqref{insidedet2} is a sum over the following directed graphs. Each connected component is a cycle (or loop). Each edge $(p_1,p_2)$ contributes a factor $M_{p_1,p_2}$  and the sign is $-1$ to the numbers of  cycles. By the definition \eqref{eq:formal-product} of the product of power series, the $n$-th coefficient of~\eqref{insidedet2} at $(x_1,\ldots, x_n)$ is 
\be \label{cycweight}
	\sum_{\sigma\in \mathfrak S_m} (-1)^{s(\sigma)} \sum_{(J_1,\ldots, J_m)} \prod_{i=1}^m \left(
			A_{\#J_{i}}(p_i;p_{\sigma^{-1}(i)},\vect{x}_{J_{i}}) \right). 
\ee
The sum is over ordered set partitions $(J_1,\ldots, J_m)$ of $[n]$ into $m$ sets. Because of $A_0=0$, only non-empty sets $J_i \neq \varnothing$ contribute.

The expression~\eqref{cycweight} is a sum over permutations $\sigma \in \mathfrak S_m$ and sequences $(G_1,\ldots, G_m)$ of spikes with vertex sets $\{\circ_i \} \cup J_i$ with disjoint $J_i$'s and base edge  $(\circ_{\sigma(i)}, \circ_{i})$. Such a sequence is naturally associated with the  digraph $G$ with vertices $\{1,\ldots, n\} \cup \{ \circ_1 , \ldots , \circ_m\}$.  Each connected component of $G$ is a crown, the base of the crown consists only out of vertices from  $\{ \circ_1 , \ldots , \circ_m\}$ and all vertices not in the base are from $\{1,\ldots, n\}$. The base of the crown corresponds to a cycle in $\sigma$. The mapping $(\sigma, (G_1,\ldots, G_m))\mapsto G$ induces  a one-to-one correspondence between (i) pairs $(\sigma, (G_1,\ldots, G_m))$ and (ii) collections of crowns with the vertices in the bases of the crown are the  vertices 
$\{ \circ_1 , \ldots , \circ_m\}$ and all of them are in some base of a crown.

Hence the $n$-th coefficient of \eqref{eq:detTK} is given by
\[
 \sideset{}{'}\sum_{\bar f\in \mathcal M[m+n]} (-1)^{s(f)} w(\bar f; \vect p_{[m]}\vect x_{[n]}), 
\]
where the primed sum is over enriched endofunctions on $[k]$ whose connected components are crowns. The next step is to drop the distinction betweed $\vect p$ and $\vect x$, which gives rise to a binomial coefficient $\binom{n+m}{m}$ which cancels with the factorial pre-factors, that is $1/m!$ in \eqref{eq:fredformal},  $1/n!$ from \eqref{eq:aqdev} and $k!=(n+m)!$ from \eqref{dkdef}.
The $k$-th coefficient of $\det (\mathrm{id} - M(z))$ is therefore equal to 

\be \label{eq:det4}
	D_k(\vect x_{[k]}) =\sum_{m=1}^k  \frac{k!}{(k-m)! m!}   \sideset{}{'}\sum_{\bar f\in \mathcal M[k]} (-1)^{s(f)} w(\bar f; \vect x)  
\ee
\end{proof}

\begin{proof} [Proof of Proposition~\ref{prop:magicformula}]
	Let $E_k^n$ and $D_k$ be the coefficients of the exponential and determinant as in~\eqref{ekdef} and~\eqref{dkdef}. 
	By~\eqref{eq:formal-product}, the power series in braces in Proposition~\ref{prop:magicformula} has as $k$-th coefficient
\be
	\sum_{I\subset [k]}  E^n_{\# I}(\vect{q};\vect x_I) D_{\# [k]\setminus I} (\vect x_{[k]\setminus I}).
\ee
In view of~\eqref{fvardev}, we have to show, in order to prove Proposition~\ref{prop:magicformula}, that 
\be \label{magickey}
	\sum_{I\subset [n]}  E^n_{\# I}(\vect{q};\vect q_I) D_{\#[n]\setminus I} (\vect q_{[n]\setminus I})= 0.
\ee
Note that $\vect x_I$ is replaced by $\vect q_I$ and $k$ has to be equal to $n$. Hence the list $\vect q_I$ is a part of $\vect q$.
	
	By Lemmas~\ref{lem:ek} and~\ref{lem:dk}, the left-hand side of~\eqref{magickey} is given by 
	\be
		\sum_{I\subset [n]} \sum_{r=1}^{n-\#I}  \sum_{\substack{\{V_1,\ldots,V_r\} \\ \text{part. of }[n]\setminus I}} \sum_{\bar f\in \mathcal{M}^{[n]\setminus I}(I)} \sum_{\substack{(\bar g_1,\ldots,\bar g_r)\\ \in \mathcal{K}(V_1)\times \cdots \times\mathcal{K}(V_r)}} w(\bar f;\vect{q}_I) \prod_{k=1}^r \bigl( - w(\bar g_k;\vect{q}_{V_k})\bigr).
	\ee
	In words: a sum over collections of crowns $\bar g_1,\ldots,\bar g_r$ with disjoint supports $V_1,\ldots, V_r$ and a forest $\bar f$ of cycle-rooted trees with root cycle in $I\subset[n]\setminus \cup_k V_k$ and vertex-rooted trees  with root vertices
	in $\cup V_k$.  
	Each such tuple can be mapped to a map $\bar F\in \mathcal{M}[n]$. The relation between weights is 
	\be
		w(\bar F;\vect{q})=w(\bar f;\vect{q}_I) \prod_{k=1}^r w(\bar g_k;\vect{q}_{V_k}). 
	\ee
	The mapping $(\bar f,\bar g_1,\ldots, \bar g_r)\mapsto \bar F$ is surjective but not injective: each $\bar F\in\mathcal{M}[n]$ is a forest of cycle-rooted trees (no vertex-rooted tree). Each such cycle-rooted tree can decide whether (a) it belongs to the forest $\bar f$, or (b) it is split into a crown $\bar g$ and attached vertex-rooted trees. 
Note that the set of splitting points is uniquely determined. The two choices for a given $\bar F$ come
with opposite signs and sum to zero.
\end{proof}

\subsection{Forests with several sinks. Proof of Theorem~\ref{thm:lagrange-good}}


The proof of Theorem~\ref{thm:lagrange-good} builds on the cancellations from Proposition~\ref{prop:magicformula}. The directed graphs in the proof of Proposition~\ref{prop:magicformula} consists of crowns and vertex-rooted trees which can be combined in such a way that all connected components are cycle-rooted trees. In contrast, for Theorem~\ref{thm:lagrange-good} the graphs  consist of crowns and two types of vertex-rooted trees. The first  type is as in Proposition~\ref{prop:magicformula}, the second one is related to the coefficients of the power series of $\Phi$. As in Proposition~\ref{prop:magicformula} the cycle-rooted trees cancel exactly with the first type of vertex-rooted trees and only the second type of trees survives, which is a combinatorical justification that the inversion is related to vertex-rooted trees only.

Let us introduce a shorthand for a forest of trees
\be \label{eq:forestsdef}
		 F_{\#L}^n(\vect q_{[n]};\vect q_L) 				:= \sum_{(V_\ell)_{\ell \in L}} \prod_{\ell\in L} t_{\#V_\ell} \bigl( q_\ell; (q_j)_{j\in V_\ell}\bigr)
	\ee
	where is the sum is over ordered partitions $(V_\ell)_{\ell \in L}$ of $[n]\setminus L$, with $V_\ell = \varnothing$ allowed.

\begin{lemma}\label{lem:cancel}
Let $n\in \N$ and $(q_1,\ldots, q_n) \in \mathbb X^n$. Fix $L\subset [n]$ and consider sums over pairs $I,J \subset[n]$ such that $[n]$ is the disjoint union of $L$, $I$, and $J$. Then 
	\be \label{eq:magicforests}
		 E_{\#I}^n(\vect q_{L\cup J \cup I};\vect q_I) = \sum_{(I_1,I_2)} F_{\#L}^{\#L + \#I_1}(\vect q_{L};\vect q_{I_1}) E_{\#I_2}^{\#J + \#I_2}(\vect q_{J \cup I_2};\vect q_{I_2}),	
	\ee 
where $I_1$, $I_2$ are disjoint subsets of $I$ such that $I_1 \cup I_2 =I$ (the sets $I_1, I_2$ could also be empty).	
\end{lemma}

\begin{proof} 
By Lemma~\ref{lem:ek} we have that
\be
E_{\#I}^n(\vect q_{L\cup J \cup I};\vect q_I) = \sum_{\bar f\in  \mathcal{M}^{L \cup J }(I)} w(\bar f;q_1,\ldots, q_n).
\ee
Define $I_1$ as the union of all the vertices in $I$  from connected components which are vertex rooted trees with root vertex in $L$. Then $I_2$ is the set of all vertices in $I$ associated to root vertices in $J$. This correspondence is one to one.
\end{proof}

\begin{proof}[Proof of Theorem~\ref{thm:lagrange-good}]
	By the product formula~\eqref{eq:formal-higher-product}, the $n$-th coefficient of 
	\be
		\Phi(z) \exp\Bigl( \sum_{i=1}^n A(q_i;z)\Bigr)\det \left( \mathrm{Id} 
		- z(\dd q)\frac{\delta}{\delta z(q')} A(q;z) \right)
	\ee
	at $(q_1,\ldots, q_n)$ is 
	\be \label{eq:magic3}
		\sum_{(L,I,J)} \Phi_{\#L}(\vect q_L)  E_{\#I}^n(\vect q_{[n]}; \vect q_I) D_{\#J}(Q;\vect q_J),
	\ee
	where the sum is over ordered partitions $L,I,J$ of $[n]$ with empty sets allowed. By the definition~\eqref{eq:zetadef} of $\zeta(\dd q;\nu)$ and the definition~\eqref{fcomp2} of the composition, the $n$-th coefficient of $\Psi(\nu) = \Phi(\zeta (\nu))$ is 
	\be \label{eq:magic4}
		\sum_{L\subset [n]} \Phi_{\#L}(\vect q_L) \sum_{(V_\ell)_{\ell \in L}} \prod_{\ell\in L} t_{n_\ell -1} \bigl( q_\ell; (q_j)_{j\in V_\ell}\bigr)=
		\sum_{L\subset [n]} \Phi_{\#L}(\vect q_L)
		F_{\#L}^n(\vect q_{[n]};\vect q_L) ,
	\ee
	where the sum is over ordered partitions $(V_\ell)_{\ell \in L}$ of $[n]\setminus L$, with $V_\ell =\varnothing$ allowed. So we have to show that the expressions~\eqref{eq:magic3} and~\eqref{eq:magic4} are equal. Indeed, by Lemma~\ref{lem:cancel} we have that 
\begin{align*}
& \sum_{(L,I,J)} \Phi_{\#L}(\vect q_L)  E_{\#I}^n(\vect q_{[n]}; \vect q_I) D_{\#J}(\vect q_J)\\
 & = \sum_{(L,I,J)} \Phi_{\#L}(\vect q_L)  \sum_{(I_1, I_2) \, : \, I_1 \cup I_2 = I}F_{\#L}^{\# L + \#I_2}(\vect q_{L\cup I_2};\vect q_{I_2}) E_{\#I_1}^{\#I_1 + \#J}(\vect q_{ I_1 \cup J }; \vect q_{I_1}) D_{\#J}(\vect q_J) .
 \end{align*}
 Using that $L, J, I_1, I_2$ are disjoint with union $[n]$, we obtain that
\begin{align*}
 & = \sum_{(L,I_2)} \Phi_{\#L}(\vect q_L)F_{\#L}^n(\vect q_{L\cup I_2};\vect q_{I_2})  
\sum_{(I_1,J)}  E_{\#I_1}^{\#J + \#i_1}(\vect q_{ I_1 \cup J }; \vect q_{I_1})  D_{\#J}(\vect q_J).
\end{align*}
Using the cancelations from Proposition~\ref{prop:magicformula} (most easily in the form \eqref{magickey}) we get that 
\[
= \sum_{(L,I_2)} \Phi_{\#L}(\vect q_L)F_{\#L}^n(\vect q_{L\cup I_2};\vect q_{I_2}),  
\]
where $L$ and $I_2$ are disjoint with union $[n]$. This is exactly  \eqref{eq:magic4}.
\end{proof} 

\section{Discussion} \label{sec:discussion}

In this final section we discuss connections to similar methods in different contexts. Some occurrences of trees in various areas of mathematics are listed in Table~\ref{table:trees}.

\begin{table}
	\begin{tabular}{ll}
		\hline \hline
		Algebra & Bass-Connell-Wright tree formula for reverse power series \cite{bass-connell-wright1982}\\
		Numerics & Butcher trees~\cite{butcher1972}\\
		Dynamical systems & Lindstedt series in KAM theory\\
		Analysis & Fa{\`a} di Bruno formula for derivatives of $f_1\circ f_2\circ\cdots \circ f_m$ \cite{johnson-prochno2019}\\
		Gaussian fields, QFT & Perturbative expansion of one-point correlation functions\\
		QFT, renormalization & Gallavotti-Niccol{\`o} trees, Hepp sectors \\
		Combinatorics	 &  Implicit species \cite[Theorems 3.1.2 and~3.2.1]{bergeron-labelle-leroux1998book}\\
		Probability theory & Branching processes \cite{good1960, good1965}\\
			&  Random trees and random forests \cite[Chapter 6]{pitman2002book} \\
		\hline\hline 
		\end{tabular} 
		\label{table:trees}

		\begin{caption}
					{Some occurrences of trees in different areas of mathematics. The abbreviations KAM and  QFT stand Kolmogorov-Arnold-Moser and quantum field theory, respectively.}
		\end{caption}
\end{table}

\subsection{Gallavotti trees}

For simplicity let us consider a gas consisting of classical particles interacting via a two-body potential $V$ at inverse temperature $\beta$. From basic statistical mechanics, we can write the density $\rho$ as a function of the activity $z$ as follows:
\begin{equation}\label{z_rho}
	\rho(z)= z+\sum_{n\geq 2} n b_n z^n,
\end{equation}
where $b_n$ are the Mayer coefficients related to the pair potential $V$.
We observe that one can invert the above expression following different strategies; we present a few.
From \eqref{z_rho}, solving for $z$ we have:
\begin{equation}\label{inv1}
	z=\rho(z)-\sum_{n\geq 2} n b_n z^n
\end{equation}
By iterating over $z$ and expanding the powers of the sums one obtains terms either with $\rho(z)$
or with $z$ in which case we keep expanding.
It is easy to visualize this procedure: each iteration is a branching of a tree and each vertex of the children either has a $\rho(z)$ (and we stop) or it has a $z$ and we continue to the next generation.
Hence, overall this expansion can be viewed as a power series ,in $\rho(z)$ with the power representing the number of final points.
This construction can be also found in 
in the Main Theorem~3 in \cite{wright1989treeformula} and in \cite{bass-connell-wright1982}, see also the survey \cite{wright2000}.
We note that this method of inverting power series has been already used in statistical mechanics,
it is actually reminiscent of Gallavotti's approach to express the Lindstedt perturbation series in the context of KAM theory \cite{gallavotti2012perturbation}.


We observe that in the above example trees are generated by iterating 
the mapping $z\mapsto \rho-\sum_{n\geq 2} n b_n z^n$ and they provide a power series
representation of the fixed point solution of the mapping.
We notice that other mappings can be suggested and, as it is usually the case, they 
correspond to more or less efficient methods. More precisely, some alternative mappings are
\begin{equation}\label{map2}
z\mapsto\frac{\rho}{1-\sum_{n\geq 2} n b_n z^n}
\end{equation}
and
\begin{equation}\label{map3}
z\mapsto\rho\e^{A(z)},\qquad\text{where}\quad A(z)=\sum_{n\geq 1}\frac{a_n}{n!}z^n,
\end{equation}
for some coefficients $a_n$. In \cite{jansen-kuna-tsagkaro2019virialinversion} we demonstrated, cf.  
Theorem 4.1, that \eqref{map3} has a  much better radius of convergence at least in some regimes.

\subsection{Abdesselam's approach}
In \cite{abdesselam2003physicistproof}, A. Abdesselam presents an alternative proof of Lagrange-Good multivariable inversion formula using a quantum field theory (QFT) model. Details are given in the companion papers \cite{abdesselam2003feynman} and \cite{abdesselam2003jacobian}.
This reveals an interesting connection between QFT calculations and Gessel's combinatorial proof \cite{gessel1987}, and seems to show a similar kind of cancellations. One ingredient of Abdesselam's proof is a representation of the one-point correlation function of some complex bosonic field as a sum over trees, which connects to the representation of the inverse in terms of trees. Another ingredient is a graphical representation of a calculus of formal power series, coming with an algebraic formalization of Feynman diagrams \cite{abdesselam2003feynman}.  

\subsection{Other connections}
Tree expansions is a favourite topic in several areas of mathematics. Without pretending of being exhaustive, we note Butcher series in computing higher order Runge-Kutta methods  \cite{butcher1972,faris2019butchertrees}
and the
combinatorial structure in indexing Hepp sectors in renormalization and regularity structures 
\cite{hairer2018bphz}. 
A common feature is that they provide a power series representation of the solution of a fixed point problem and as such 
we believe expansion methods of this type can be widely used in applications. The techniques developed in this paper can be used to extend these expansion in an infinite dimensional context, for example in an inhomogeneous situation as in this paper.

\appendix

\section{Formal power series} \label{ApFormal}

Here we describe some operations on  formal power series as introduced in Definition~\ref{def:fps}, e.g., \be \label{eq:formal}
	K(z) = K_0 + \sum_{n=1}^\infty \frac{1}{n!} \int_{\mathbb X^n} K_n(x_1,\ldots,x_n)  z(\dd x_1)\cdots  z(\dd x_n)
\ee
where $(\mathbb X,\mathcal{X})$ is a measurable space $z$ is a measure on $(\mathbb X,\mathcal X)$, and $K_0\in \C$ is a scalar, and $K_n:\mathbb X ^n \to \C$ are measurable maps that are invariant under permutation of the arguments. Operations are defined purely in terms of the sequence of coefficients.

\noindent \emph{Product}.
Let $K,G$ be formal power series, then $KG$ is defined by 
\be \label{eq:formal-product}
	(KG)_n (x_1,\ldots,x_n):=\sum_{\ell=0}^n \sum_{J\subset [n], \#J=\ell}  K_\ell\bigl( (x_j)_{j\in J} \bigr)G_{n-\ell}\bigl( (x_j)_{j\in [n]\setminus J}\bigr).
\ee
For a motivation of this definition, see e.g.~\cite[Appendix A]{jansen-kuna-tsagkaro2019virialinversion}.
The empty set $J=\varnothing$ is explicitly allowed. As an operation on sequences of symmetric functions, this is exactly the convolution in~\cite[Chapter 4.4]{ruelle1969book}.
 It is not difficult to check that the product is commutative and associative. 
  Eq.~\eqref{eq:formal-product} generalizes to products $K^{(1)}\cdots K^{(r)}$ as 
\be \label{eq:formal-higher-product}
	\bigl(K^{(1)}\cdots K^{(r)}\bigr)_n(x_1,\ldots,x_r)
		= \sum_{(V_1,\ldots, V_r)} \prod_{\ell =1}^r K^{(\ell)}_{\#V_\ell}\bigl( (x_j)_{j\in V_\ell}\bigr)
\ee
where the sum runs over ordered partitions $(V_1,\ldots, V_r)$ of $[n]$ into $r$ disjoint parts, with $V_i=\varnothing$ explicitly allowed.  \\

\noindent {\emph{Variational derivative}}. For $q\in \mathbb X$ and $K$ a formal power series over $\mathbb X$, we define 
\be \label{fvardev}
	\Bigl(\frac{\delta }{\delta z(q)} K\Bigr)_n (x_1,\ldots,x_n) \equiv 	\Bigl(\frac{\delta K}{\delta z}\Bigr)_n (q;x_1,\ldots,x_n)= K_{n+1}(q,x_1,\ldots,x_n). 
\ee
For higher order variational derivatives, see Definition~\ref{def:varderivative}. 
\\

\noindent \emph{Integrals. Measures $K(\dd q;z) = F(q;z) z(\dd q)$}.  Let $F(q;z)$ be a function-valued power series. The power series 
\be
	I(z)= \int_\mathbb X F(q;z) z(\dd q)
\ee
is defined as the power series with coefficients $I_0:=0$ and 
\be \label{eq:formal-integral}
	I_n(x_1,\ldots, x_n):= \sum_{r=1}^n F_n\bigl(x_r; (x_j)_{j\in [n]\setminus \{r\}}\bigr). 
\ee
The definition is motivated by the following formal computation: 
\begin{align*}
	\int_{\mathbb X} F(q;z) z(\dd q) & = \sum_{m=0}^\infty \frac{1}{m!}\int_{\mathbb X^{m+1}} F_m\bigl( q;\vect x_{[m]}\bigr) z^m(\dd \vect x) z(\dd q) \\
	& =  \sum_{m=0}^\infty \frac{1}{m!}\int_{\mathbb X^{m+1}} F_m\bigl( x_{m+1};\vect x_{[m]}\bigr) z^{m+1}(\dd \vect x)  \\
	& =  \sum_{m=0}^\infty \frac{1}{m!}\int_{\mathbb X^{m+1}} \frac{1}{m+1} \sum_{r=1}^ {m+1}F_m\bigl( x_{r};\vect x_{[m+1]\setminus \{r\}}\bigr) z^{m+1}(\dd \vect x) \\
	& =  \sum_{n=1}^\infty \frac{1}{n!}\int_{\mathbb X^{n}} \sum_{r=1}^{n}F_n\bigl( x_{r};\vect x_{[n]\setminus \{r\}}\bigr) z^{n}(\dd \vect x).	
\end{align*} 
We also define the measure-valued formal power series $K(\dd q;z) = z(\dd q) F(q;z)$ as the power series with coefficients $K_0:=0$ and 
\be \label{eq:formal-radon}
	K_n(\mathrm d q;x_1,\ldots, x_n):= \sum_{r=1}^n \delta_{x_r}(\mathrm d q) F_n\bigl(x_r; (x_j)_{j\in [n]\setminus \{r\}}\bigr). \\
\ee

%

\noindent \emph{Composition I and exponential series.}
Let $F(t) = \sum_{n=0}^\infty f_n t^n/n!$ be a formal power series in a single variable $t$ and $K$ a formal power series on $(\mathbb X,\mathcal X)$ with $K_0=0$. The formal power series $F\circ K$ on $\mathbb X$ is defined by $(F\circ K)_0:=f_0$ and for $n\geq 1$, 
\be \label{composition2} 
	(F\circ K)_n (x_1,\ldots,x_n) := \sum_{m=1}^n\sum_{\{J_1,\ldots,J_m\}\in \mathcal P_n} f_{m} \prod_{\ell=1}^m  K_{\#J_\ell}\bigl( (x_j)_{j \in J_\ell}\bigr)
\ee
with $\mathcal P_n$ the collection of set partitions of $\{1,\ldots,n\}$. Note that only because $K_0=0$ the expression \eqref{composition2} is well-defined as a formal power series, because only in this case the sum is finite. An important special case is $F(t) = \exp(t)$, for which Eq.~\eqref{composition2} becomes 
\be \label{eq:formal-exponential} 
	(\exp(K))_n (x_1,\ldots,x_n) = \sum_{m=1}^n\sum_{\{J_1,\ldots,J_m\}\in \mathcal P_n}  \prod_{\ell=1}^m  K_{\#J_\ell}\bigl( (x_j)_{j \in J_\ell}\bigr),
\ee
which is exactly the exponential on the algebra of symmetric functions from~\cite[Chapter 4.4]{ruelle1969book}.  For a motivation of this definition, see again e.g.~\cite[Appendix A]{jansen-kuna-tsagkaro2019virialinversion}.
\\

\noindent \emph{Composition II}.  In order to define the compositions in ~\eqref{eq:inverse}, we need a more general type of composition.  Let $G$ be a formal power series on $\mathbb X$  and $F(q;z)$ a function-valued power series 
$$
	F(q;z) =F_0(q)+ \sum_{n=1}^\infty \frac{1}{n!}\int_{\mathbb X^n} F_n(q;x_1,\ldots,x_n) z(\dd x_1) \cdots z(\dd x_n). 
$$
Let $K(\dd q;z)$ be the measure-valued formal power series $K(\dd q;z) = F(q;z)z(\dd q)$ with coefficients~\eqref{eq:formal-radon}. The composition $H(z):= (G\circ K)(z)$ is defined as the power series with coefficients $H_0=0$ and for $n\geq 1$, 
\be \label{fcomp2}
	H_n(x_1,\ldots, x_n)  = \sum_{m=1}^n \sum_{\substack{J\subset [n]  \\ \#J =m}}  G_{m}\bigl((x_j)_{j\in J}\bigr) \sum_{\substack{(V_j)_{j\in J}:\\  \dot \cup_{j\in J} V_j  = [n]\setminus J}} \prod_{j\in J} F_{\# V_j}\bigl(x_j; (x_v)_{v\in V_j}\bigr).
\ee
The summation is over partitions $(V_j)_{j\in J}$ of $[n]\setminus J$ with empty sets $V_j = \varnothing$ allowed. Note that the sum is only over finitely many summands and hence well-defined as a formal power series. The definition is motivated by the following formal computation: 
\begin{align*}
	H(z) & = G_0 + \sum_{m=1}^\infty\frac{1}{m!} \int_{\mathbb X^m} G_m(x_1,\ldots,x_m) F(x_1;z)\cdots F(x_m;z) z(\dd x_1)\cdots z(\dd x_m) \\
	& = G_0 + \sum_{m=1}^\infty \frac{1}{m!} \int_{\mathbb X^m}G_m(x_1,\ldots,x_m) \prod_{j=1}^m \left( \sum_{n_j=0}^\infty \frac{1}{n_j!} \int_{\mathbb X^{n_j}} F_{n_j}(x_j;\vect y) z^{n_j}(\dd \vect y)\right)  z^m(\dd \vect x).
\end{align*} 
We group terms with the same sum $n= m+ n_1+\cdots + n_m$, write 
\be
	\frac{1}{m!}\prod_{j=1}^m \frac{1}{n_j!} = \frac{1}{n!} \binom{n}{m,n_1,\ldots,n_m} 
\ee
and note that the multinomial counts the number of ways to partition $[n]$ into sets $J$ and $(V_j)_{j\in J}$ in such a way that $\#J = m$ and $\#V_j = n_j$. Exploiting the symmetry of the integrands, we arrive at 
\be
	H(z)  = \sum_{n=1}^\infty \frac{1}{n!}\int_{\mathbb X^n} H_n(x_1,\ldots, x_n) z^n(\dd \vect x)
\ee
with $H_n$ defined in~\eqref{fcomp2}. \\

\noindent \emph{Composition of two measure-valued series.} The composition~\eqref{fcomp2} easily extends to the composition of a measure-valued formal power series  $\Gamma$ with a formal power series  of the form $K(\dd q;z) = F(q;z) z(\dd q)$, for which we have defined the coefficients in  Eq.~\eqref{eq:formal-radon}. 

Then $H= \Gamma \circ K$ is the measure-valued series which has by Eq.~\eqref{fcomp2} the coefficients for $n\geq 1$ and $B\subset \mathbb X$, 
\be \label{eq:fmc}
	H_n(B;x_1,\ldots, x_n):=  \sum_{\substack{J\subset [n]:\\ J\neq \varnothing}} \Gamma_{\# J}(B; \vect x_J)  \sum_{\substack{(V_j)_{j\in J}:\\  \dot \cup_{j\in J} V_j  = [n]\setminus J}} \prod_{j\in J} F_{\# V_j}\bigl(x_j; (x_v)_{v\in V_j}\bigr).
\ee
In treating the the compositions $\rho \circ \zeta$ and $\zeta \circ \rho$ in Eq.~\eqref{eq:inverse}, we need to consider $\Gamma$ which are of the form $
\Gamma(\dd q;z) = G(q;z) z(\dd q)$ for some function valued formal power series $G$, that is by Eq.~\eqref{eq:formal-radon} again
\be
	\Gamma_0(\dd q)= 0 ,\quad \Gamma_n(B;x_1,\ldots, x_n) = \sum_{r=1}^n \1_B(x_r) G_{n-1}\bigl(x_r;(x_j)_{j \neq r}\bigr).
\ee
Then also $H$ is of the form $H(\dd q;z) = J(q;z) z(\dd q)$, where $J$ is a function valued formal power series and 
as above, 
\be
	H_n(B;x_1,\ldots, x_n) = \sum_{r=1}^n \1_B(x_r) J_{n-1}\bigl(x_r; (x_j)_{j\neq r}\bigr)
\ee
with $ H_0 =0$. Then using Eq.~\eqref{eq:formal-radon} we can directly express 
\be
	J_n(q;y_1,\ldots, y_n) = \sum_{\substack{J\subset [n]:\\ J\neq \varnothing}} G_{\# J}(q; \vect y_J)  \sum_{\substack{(V_j)_{j\in J}:\\  \dot \cup_{j\in J} V_j  = [n]\setminus J}} \prod_{j\in J} F_{\# V_j}\bigl(y_j; \vect y_{V_j}\bigr) .
\ee
This also implies that 
\be \label{eq:mcomm}
	H(\dd q; z) = (\Gamma \circ K)(\dd q;z) = (G\circ K)(q;z) z(\dd q).
\ee

\noindent \emph{Restriction of power series} Let $\mathbb{Y} \subset \mathbb{X}$, then one can define a restriction of the formal power series to the set $\mathbb{Y}$,  denoted by $K\!\!\upharpoonright_{ \mathbb{Y}}$, in the following way: restrict the coefficients of $K$ to $\mathbb{Y}$, that is, consider 
\[
K_n\!\!\upharpoonright_{ \mathbb{Y}} : \mathbb{Y}^n \rightarrow \mathbb{C}, \quad (x_1, \ldots , x_n) \mapsto K_n(x_1, \ldots , x_n).
\]
One can consider $K_n\upharpoonright_{ \mathbb{Y}} $ also as a function on $\mathbb{X}^n$ by defining $K_n\upharpoonright_{ \mathbb{Y}} =0$ outside of
$\mathbb{Y}^n$ and hence we may formally write
\begin{multline} 
	K(z)\!\!\upharpoonright_{ \mathbb{Y}}  = K_0 + \sum_{n=1}^\infty \frac{1}{n!} \int_{\mathbb X^n} K_n(x_1,\ldots,x_n) \prod_{i=1}^n \1_{\mathbb Y}(x_i) z(\dd x_1)\cdots  z(\dd x_n)\\ 
	=: K_0 + \sum_{n=1}^\infty \frac{1}{n!} \int_{\mathbb Y^n} K_n(x_1,\ldots,x_n)  z(\dd x_1)\cdots  z(\dd x_n) .
\end{multline}
In case that the formal power series is an actual convergent power series, the above  restriction corresponds to restricting the function $ z \mapsto K(z)$ to all measures which are zero outside of $\mathbb{Y}$. In particular, when $\mathbb{Y}$ contains only finitely many elements, then any measure zero outside of $\mathbb{Y}$ is of the form $z(\dd x) = \sum_{y \in \mathbb{Y}} z_y\delta_{y}(\dd x)$ for some $ z_y \in \mathbb{R}_+$ and hence, in this case,  $K\upharpoonright_{ \mathbb{Y}} $ can be seen as a function on $\mathbb{C}^{\# \mathbb{Y}}$. The analogous construction works also for  $\mathbb{Y}$ with countable many elements, but not for uncountable many elements.

In order to compute the $n$-th. coefficient of $K$ evaluated at $q_1, \ldots , q_n \in \mathbb{X}$, that is $K_n(q_1,\ldots,q_n)$, it is sufficient to consider $K\upharpoonright_{ Q}$, where $Q=\{q_i:\, i=1,\ldots,n\}$ is the set of colors in $q_1, \ldots , q_n$. Notice that the set of colors $Q$ has cardinality smaller than $n$ if colors are repeated in the vector, i.e., $q_i = q_j$ for some $i\neq j$.  
The relation between the coefficients is the following 
\[
\frac{\delta^n}{\delta z(q_1)\cdots \delta z(q_n)} K(z) =  [\vect z^{\vect n}] K\upharpoonright_{ Q}(\sum_{q \in Q} z_q \delta_q),
\]
where $\vect z := ( z_q)_{q\in Q}$, $\vect n := (n_q)_{q \in Q}$, and $n_q:= \# \{ i \in \{ 1, \ldots , n \} \, : \, q_i =q \}$, that is the number of repetitions of the color $q_i$.

This allows us to reduce the computation of the Fredholm determinant, which appears in Theorem~\ref{thm:lagrange-good}, to the computation of  usual determinants
\begin{lemma} \label{lem:smalldet}
	Let $n\in \N$ and $(q_1,\ldots,q_n)\in \mathbb X^n$. Set $Q=\{q_i:\, i=1,\ldots,n\}$. Then the 
	$n$-th coefficient of the Fredholm determinant $\det (\mathrm{Id}- \mathbb K_z)$, evaluated at $(q_1,\ldots,q_n)$, is equal to the $n$-th coefficient at $(q_1,\ldots,q_n)$ of the $(\#Q)\times (\#Q)$-matrix 
\be \label{eq:smalldeterminantb} 
	z \mapsto \mathrm{det}\left(\Bigl( \delta_{q,q'} - z(\{q\}) K_z(q',q)\Bigr)_{q,q'\in Q}\right).
\ee
\end{lemma} 

\noindent

\begin{proof}
 We start with the more intuitive case of finite color set $\mathbb X = \{1,\ldots,\ell\}$ with $\ell\in \N$. In this case  the Fredholm determinant is just the determinant of an $\ell\times \ell$ matrix, 
\[
	\det(\mathrm{Id} - \mathbb K_z) = \det \Bigl( \bigl( \delta_{x,y} - z(x) K_z(y,x)\bigr)_{x,y=1,\ldots,\ell} \Bigr).
\] 
The analogue of~\eqref{eq:fredformal1aa} reads 
\[
	\det(\mathrm{Id} - \mathbb K_z) = 1+ \sum_{r=1}^\infty \frac{(-1)^r}{r!}\sum_{(x_1,\ldots,x_r)\in \mathbb X^r} \det\left( (K_z(x_i,x_j))_{i,j=1,\ldots, r}\right) z_{x_1}\cdots z_{x_r}.  
\] 
Suppose we want to know the coefficient of some monomial $z_{1}^{n_1}\cdots z_\ell^{n_\ell}$ in the determinant. Then, clearly the only relevant contributions are from summands $(x_1,\ldots, x_r)$ with every entry $x_i$ contained in the support $\mathrm{supp}\, \vect n = \{x:\, n_x \geq 1\}$, which plays the role of $Q$. The series is actually a sum, as all terms for $r > \#Q$ are zero. Noticing that
\begin{multline} \label{eq:smalldet1}
	 1+ \sum_{r=1}^\infty \frac{(-1)^r}{r!}\sum_{x_1,\ldots,x_r \in \mathrm{supp}\, \vect n} \det\left( (K_z(x_i,x_j))_{i,j=1,\ldots, r}\right) z_{x_1}\cdots z_{x_r}\\
  =	 \det \Bigl( \bigl( \delta_{x,y} - z(x) K_z(y,x)\bigr)_{x,y\in \mathrm{supp}\, \vect n}\bigr) \Bigr),
\end{multline}
we conclude
\be \label{eq:smalldet0}
	[\vect z^{\vect n}] \det(\mathrm{Id} - \mathbb K_z)
	= [\vect z^{\vect n}] \det\left( \left( \delta_{x,y} - z(x) K_z(y,x)\right)_{x,y\in \mathrm{supp}\, \vect n}\right). 
\ee
Turning back to general sets $\mathbb X$, the result follows as the coefficient only depends on $\det(\mathrm{Id} - \mathbb K_z) \upharpoonright_{Q}$ as discussed before the lemma.

\end{proof}

\noindent \emph{Determinants  for $z$ either with  density or which are generalized functions.}
In several applications, it is not natural to restrict $z$ to points. In statistical mechanics for example, it is typical that $\mathbb{X} \subset \mathbb{R}^d$ and $z$ has a density with respect to the Lebesgue measure. For such measures $z$ the restriction to $Q$ is always zero. In quantum field theory, one considers $z$ which are only generalized functions and the restriction to points has no sense at all. However, one can re-interpret the Fredholm determinant in another way and give an expression in terms of usual determinants. Assume that $z$ either has a density with respect to a reference measure $m$ or that $z$ is a generalized function. The density we denote as well by $z$ and the duality between test- and generalized functions we formally write as an integral. The $n$-th. coefficient of \eqref{eq:fredformal1aa} depends only on 
\be \label{eq:fredformal1a}
1 + \sum_{r=1}^n \frac{(-1)^r}{r!}\int_{\mathbb X^r} \det\left( \left( K_z(x_i,x_j) \right)_{i,j=1, \ldots ,r} \right) z( x_1) \ldots z( x_r) m(\dd x_1) \ldots m(\dd x_r),
\ee
where  $z( q_1) \ldots z( q_n)$ are interpreted as the $n$-fold tensor product of generalized functions. This shows that one can get the $n$-th. coefficient by just computing the determinant of an $n \times n$-matrix
\[
\int_{\mathbb X^n} \mathrm{det}\left(\Bigl( \delta(x_i - x_j) - z(x_j) K_z(x_i, x_j)\Bigr)_{i,j=1, \ldots n}\right) m(\dd x_1) \ldots m(\dd x_r) .
\]
The determinant is well-defined, because in all expressions the generalized functions $z$ are evaluated at different points. The Dirac deltas cause no problem, because they only lead to dropping integrals. Indeed, one just obtain \eqref{eq:fredformal1a}, which is well-defined for generalized functions $z$.

\section{Combinatorial species with uncountable color space} \label{app:coloredspecies}

Formal power series in finitely or countably many variables have a natural interpretation as exponential generating functions for labelled, colored combinatorial species, which helps prove identities of power series identities independently of any convergence considerations. This point of view is formalized with Joyal's theory of combinatorial species~\cite{joyal1981, bergeron-labelle-leroux1998book}. Power series in several variables correspond to colored combinatorial species, also called multisort species~\cite{bergeron-labelle-leroux1998book}, with one variable $z_k$ per color or sort.  See the survey by Faris~\cite{faris2010combinatorics} for an account in the context of cluster expansions and~\cite{faris2009feynman} for applications to Feynman diagrams. 

This appendix extends the theory of combinatorial species generalizes the theory of colored species to infinite, possibly uncountable color space $\col$. Such a generalization was in fact already proposed by M{\'e}ndez and Nava~\cite{mendez-nava1993}, however the concept of generating function proposed by M{\'e}ndez and Nava is too restrictive for our purpose. Indeed their generating functions are sums of monomials $\prod_{k\in \col} z_k^{n_k}$ indexed by multi-indices $\vect n = (n_k)_{k\in \col}$ that have only finitely many non-zero entries $n_k \neq 0$. Our generating functions instead are functions of  measures $z(\dd x)$ on the color space $\col$. This is the kind of generating function that appears naturally in the statistical mechanics of inhomogeneous systems, where colors may  correspond, for example,  to positions $x\in \R^d$ in space and the measure $z(\dd x)$ is a position-dependent activity \cite{stell1964}. 

Another difference with~\cite{mendez-nava1993} is a more nuanced notion of family of power series. For finite color spaces and finitely many variables, it is natural to look at families $(F_k)_{k\in C}$ of generating functions, e.g., for rooted colored trees that have their root of color $k$. In our setup the relevant power series are either function or measure-valued. For example, rooted colored trees give rise to a family $(T_B^\bullet(z))_{B\subset \col}$ indexed by sets $B\subset \col$ instead of elements $k\in \col$:   to each set of colors $B\subset \col$ associate the generating function for trees whose root has color in $B$. \\

We proceed with a short self-contained description of our formalism, which is similar to~\cite{mendez-nava1993} but has a different notion of generating function. We start with the concrete example of graphs. Let $\col$ be a non-empty possibly infinite set, the set of \emph{colors}. Let $V$ be a finite set, e.g., $V = \{1,\ldots,n\}= [n]$ with $n\in \N$; $V$ is the set of  \emph{labels}.  A $\col$-colored, labelled graph on $V$ is a pair $(G,c)$ consisting of (i) a graph $G$ with vertex set $V$, (ii) a map $c:V\to C$. Thus every vertex of the graph is assigned both a label  $v\in V$  and a color $c(v)\in C$. The pair $(V,c)$ is a \emph{$\col$-colored set}.  We write $\mathcal{G}(V,c)$ for the collection of colored labelled graphs on $V$ with prescribed coloring $c$. 

Admissible relabellings of vertices are formalized with bijections: Let $(W,\tilde c)$ be another finite colored set and $\varphi:V\to W$ a color-preserving bijection, i.e., $\tilde c(\varphi(v)) = c(v)$ for all $v\in V$. 
Relabelling the vertices $v$ of a graph $G\in \mathcal{G}(V,c)$ by $w=\varphi(v)$ we obtain a graph $\tilde G\in \mathcal{G}(W,\tilde c)$. In this way $\varphi$ induces a bijection between $\mathcal{G}(V,c)$ and $\mathcal{G}(W,\tilde c)$.  Choosing $V=W=[n]$, we deduce that for every permutation $\sigma\in \mathfrak{S}_n$, we have 
$\#\mathcal{G}([n],c) =\#\mathcal{G}([n],c\circ \sigma)$. Put differently,  $\#\mathcal{G}([n],c)$ is a symmetric function of the variables $c_i=c(i)$. 
The associated formal power series 
\be
	\sum_{n=1}^\infty \frac{1}{n!}\int_{C^n} \#\mathcal{G}([n],c) z(\dd c_1)\cdots z(\dd  c_n)
\ee
is the \emph{(exponential) generating function} of the species of labelled, $\col$-colored graphs. 

In the following we consider the set of colors $C$ to be fixed once and for all.

\begin{definition}
	A \emph{colored set} is a pair $(V,c)$ consisting of a finite set $V$ and a map $c:V\to C$. A \emph{color-preserving bijection} $\varphi:(V,c)\to(W,\tilde c)$ is a bijection from $V$ onto $W$ such that $\tilde c = c\circ \varphi$.
\end{definition}

\noindent The empty set $V=\varnothing$ is considered a colored set. This is needed for the combinatorial counterpart of the variational derivative (see below) and  allows for a conceptualization of  pinned vertices that are not integrated over in generating functions. For example, we may be interested in the set of trees with vertex set $[n]\cup\{\star,\circ\}$ where $\star,\circ\notin [n]$ are two distinct  elements not in $[n]$. Then $n=0$ corresponds to trees with vertex set $\{\star,\circ\}$.  

Sometimes we write colorings as vectors $(c_v)_{v\in V}$ and not as maps $v\mapsto c(v)$.

\begin{definition}
	A labelled combinatorial species $F$ with colors in $C$ consists of two rules:
	\begin{enumerate}
		\item a rule $(V,c)\mapsto F(V,c)$ that assigns to every colored set a finite set (for $V=\varnothing$ we write instead $F(\varnothing)$),
		\item a rule that assigns to every color-preserving bijection $\varphi: (V,c)\to (W,\tilde c)$ between colored sets a bijection $\Phi:F(V,c)\to F(W,\tilde c)$ (\emph{relabelling}),
	\end{enumerate}
	that satisfy the following: for all color-preserving bijections $\varphi_1:(V_1,c_1)\to (V_2,c_2)$, $\varphi_2:(V_2,c_2)\to (V_3,c_3)$, the relabelling induced by $\varphi_3=\varphi_2\circ\varphi_1$ is the composition of the relabellings, $\Phi_3 =\Phi_2\circ \Phi_1$. 
\end{definition} 

\noindent Put differently, a colored combinatorial species is a functor from the category of of colored sets (objects: finite colored sets, morphisms: color-preserving bijections) to the category of finite sets. In concrete examples the concept of relabelling and its functorial property are often so natural that the rule $\varphi\mapsto \Phi$ is left implicit.

\begin{definition}
	Let $F$ be a labelled colored combinatorial species. 
	Suppose that 	for each colored map $(V,c)$, we are given a weight function $w_{(V,c)}:F(V,c)\to \C$ and that for all color-preserving bijections $\varphi:(V,c)\to (W,\tilde c)$, $w_{(W,\tilde c)} \circ \Phi = w_{(V,\tilde c)}$. The pair $(F,w)$ of consisting of the species $F$ and the family $w$ of weights $w_{(V,c)}$ is a \emph{weighted colored species}. 
\end{definition} 

\noindent By a slight abuse of notation we shall omit the indices and use the letter $w$ both for the family of weight functions and for weight maps $w:F(V,c)\to \C$. 

\begin{lemma} 
	Let $(F,w)$ be a weighted colored species. Then for every $n\geq 1$, every coloring $c:[n]\to C$ and all $\sigma \in \mathfrak{S}_n$, 
	$$\sum_{g\in F([n],c)} w(g) = \sum_{g\in F([n],c\circ \sigma)} w(g).$$
\end{lemma} 

\noindent Thus $\sum_{g\in G([n],c)} w(g)$ is a symmetric function of the color variables $c(1),\ldots,c(n)$ and we may apply the notion of formal power series. The elementary proof is left to the reader. 

\begin{definition} 
	The generating function of a weighted colored species $(F,w)$ is the formal power series 
	$$
		F(z) = \sum_{g\in F(\varnothing)} w(g) +\sum_{n=1}^\infty \frac{1}{n!}\int_{C^n} \Bigl( \sum_{g\in F([n],c)} w(g)\Bigr) z(\dd c_1)\cdots z(\dd c_n). 
	$$
\end{definition} 	

%

\begin{example}[Colored singletons]
Let $B\subset \col$ be a set of colors. The species of singletons with color in $B$ is the species 
\be
	F_B(V,c) = \begin{cases}
		\{(V,c)\}, &\quad V = \{v\}\text{ is  a singleton and } c(v) \in B,\\
		0, &\quad \text{else}. 
	\end{cases}
\ee
The associated generating function is 
\be
	F_B(z) = \int_\col \1_B(c_1) z(\dd c_1)  = z(B),
\ee
compare Example~\ref{ex:z}.  
\end{example}

Operations on formal power series correspond to operations on combinatorial species. We provide formulas for non-weighted species only, the generalization to weighted species is straightforward.  \\

\noindent \emph{Cartesian product}. Let $F,G$ be two combinatorial species. We define a new species $F\times G$ by 
$$
	(F\times G)(V,c):= \bigcup_{\substack{V_1,V_2\subset V:\\ V_1\cap V_2 = \varnothing, V_1\cup V_2= V} } F(V_1,c\bigr|_{V_1}) \times G(V_2,c\bigr|_{V_2}).
$$
The generating function is $(F\times G)(z) = F(z) G(z)$, compare Eq.~\eqref{eq:formal-product}. Hence an $F\times G$-structure on $V$ is a pair $(f,g)$ consisting of an $F$-structure on $V_1$ and a $G$-structure on $V_2$, with $V_1,V_2$ a partition of $V$ into possibly empty sets. \\

\noindent
\emph{Derivatives.}
Let $(V,c)\mapsto G(V,c)$ be some colored combinatorial species. Suppose that for each finite set $V$ there is a designated element $\circ = \circ_V$ that is not in $V$, see Definition~5, Remark~6, and Exercise~16 in~\cite[Chapter~1.4]{bergeron-labelle-leroux1998book}.  Given $q\in \col$, we extend a  coloring $c:V\to \col$ to a coloring $c_q^\circ: V\cup \{\circ\} \to \col$ by setting 
\be
	c_q^\circ (\circ ) = q,\quad c_q^\circ (v) = c(v) \quad (v\in V). 
\ee
Then we define a family $(F_q)_{q\in \col}$ of colored species by 
\be
	G_q^\circ(V,c) := G\bigl( V\cup \{\circ\}, c_q^\circ\bigr).
\ee
The generating function is 
\be
	G_q^\circ(z) = \sum_{n=1}^\infty \frac{1}{n!} \int_{\col ^n}\# G\bigl( [n]\cup \{\circ\}, (c_j)_{j\in V\cup \{\circ\}}\bigr)  z(\dd c_1)\cdots z(\dd c_n),\quad c_\circ =q
\ee
in which we recognize the variational derivative 
\be
	 G_q(z) = \frac{\delta G}{\delta z}(q;z) = \frac{\delta}{\delta z(q)} G(z), 
\ee
see Eq.~\eqref{fvardev}. 
We may think of objects in $G_q^\circ$ as objects of $G$ rooted in $\circ$. For example, when $G= T$ is the species of non-rooted trees, then $T_q^\circ$ can be identified with the species of trees rooted in the non-labelled vertex (ghost)  $\circ$ 
with prescribed root color $q$. \\
%

\noindent \emph{Pointing}. Let $G$ be a species and $(V,c)$ a finite colored set. For $B\subset C$, define 
\be
	 G_B^\bullet(V,c):= \{ (g,r)\mid g\in G(V,c),\, r\in V,\, c(r) \in B\}.
\ee
For example, when $G=T$ is the species of the non-rooted trees, $T_B^\bullet$  corresponds to rooted trees with labelled root and root color in $B$. We have 
\be
	 \#G_B^\bullet([n],c) =\sum_{r=1}^n \1_B(c(r))\, \# G([n],c) .
 \ee
Comparison with Eq.~\eqref{eq:formal-radon} yields 
\be
	G_B^\bullet(z) = \int \1_B(q) G_q^\circ(z) z(\dd q). 
\ee
Compare Lemma~\ref{lem:treeformula} and the paragraph preceding the lemma. \\



\subsubsection*{Acknowledgments} 
	S.J. thanks the GSSI,  T.K. the University in L'Aquila, Italy,  and D.T. thanks the LMU in Munich, Germany, for hospitality. S.J. was supported by the Munich Center for Quantum Science and Technology (MCQST). We thank Abdelmalek Abdesselam, Giovanni Gallavotti, and Alessandro Giuliani for fruitful discussions and Nils Berglund for pointing out the reference~\cite{hairer2018bphz}.

\bibliographystyle{amsalpha}
\bibliography{virial-combinatorics}

\providecommand{\bysame}{\leavevmode\hbox to3em{\hrulefill}\thinspace}
\providecommand{\MR}{\relax\ifhmode\unskip\space\fi MR }
\providecommand{\MRhref}[2]{%
  \href{http://www.ams.org/mathscinet-getitem?mr=#1}{#2}
}
\providecommand{\href}[2]{#2}
\begin{thebibliography}{JTTU14}

\bibitem[Abd03a]{abdesselam2003feynman}
A.~Abdesselam, \emph{Feynman diagrams in algebraic combinatorics.},
  S{\'e}minaire Lotharingien de Combinatoire \textbf{49} (2003), B49c, 45 p.

\bibitem[Abd03b]{abdesselam2003jacobian}
\bysame, \emph{The {J}acobian conjecture as a problem of perturbative quantum
  field theory}, Ann. Henri Poincar\'{e} \textbf{4} (2003), no.~2, 199--215.

\bibitem[Abd03c]{abdesselam2003physicistproof}
\bysame, \emph{A physicist's proof of the {L}agrange-{G}ood multivariable
  inversion formula}, J. Phys. A \textbf{36} (2003), no.~36, 9471--9477.

\bibitem[Bar85]{barnabei1985}
M.~Barnabei, \emph{Lagrange inversion in infinitely many variables}, J. Math.
  Anal. Appl. \textbf{108} (1985), no.~1, 198--210.

\bibitem[BCLL03]{bousquet-chauve-labelle-leroux2003}
M.~Bousquet, C.~Chauve, G.~Labelle, and P.~Leroux, \emph{Two bijective proofs
  for the arborescent form of the {G}ood-{L}agrange formula and some
  applications to colored rooted trees and cacti}, Theoret. Comput. Sci.
  \textbf{307} (2003), no.~2, 277--302.

\bibitem[BCW82]{bass-connell-wright1982}
H.~Bass, E.~H. Connell, and D.~Wright, \emph{The {J}acobian conjecture:
  reduction of degree and formal expansion of the inverse}, Bull. Amer. Math.
  Soc. (N.S.) \textbf{7} (1982), no.~2, 287--330.

\bibitem[BLL98]{bergeron-labelle-leroux1998book}
F.~Bergeron, G.~Labelle, and P.~Leroux, \emph{Combinatorial species and
  tree-like structures}, Encyclopedia of mathematics and its applications,
  vol.~67, Cambridge University Press, 1998.

\bibitem[BR98]{bender-richmond1998}
E.~A. Bender and L.~B. Richmond, \emph{A multivariate {L}agrange inversion
  formula for asymptotic calculations}, Electron. J. Combin. \textbf{5} (1998),
  Research Paper 33, 4 pp.

\bibitem[But72]{butcher1972}
J.~C. Butcher, \emph{An algebraic theory of integration methods}, Math. Comp.
  \textbf{26} (1972), 79--106.

\bibitem[EM94]{ehrenborg-mendez1994}
R.~Ehrenborg and M.~M\'{e}ndez, \emph{A bijective proof of infinite variated
  {G}ood's inversion}, Adv. Math. \textbf{103} (1994), no.~2, 221--259.

\bibitem[Far11]{faris2009feynman}
W.~G. Faris, \emph{Combinatorial species and {F}eynman diagrams}, S\'{e}m.
  Lothar. Combin. \textbf{61A} (2009/11), Art. B61An, 37 pp.

\bibitem[Far10]{faris2010combinatorics}
\bysame, \emph{Combinatorics and cluster expansions}, Probability Surveys
  \textbf{7} (2010), 157--206.

\bibitem[Far19]{faris2019butchertrees}
\bysame, \emph{Rooted tree graphs and the {B}utcher group: Combinatorics of
  elementary perturbation theory}, Sojourns in Probability and Statistical
  Physics - II: Brownian Web and Percolation, A Festschrift for Charles M.
  Newman (V.~Sidoravicius, ed.), Springer Proceedings in Mathematics \&
  Statistics, vol. 299, Springer Nature, Singapore, 2019, pp.~135--166.

\bibitem[FS09]{flajolet-sedgewick2009book}
P.~Flajolet and R.~Sedgewick, \emph{Analytic combinatorics}, Cambridge
  University Press, Cambridge, 2009.

\bibitem[Gal12]{gallavotti2012perturbation}
G.~Gallavotti, \emph{Perturbation theory}, Mathematics of complexity and
  dynamical systems. {V}ols. 1--3, Springer, New York, 2012, pp.~1290--1300.

\bibitem[Ges87]{gessel1987}
I.~M. Gessel, \emph{A combinatorial proof of the multivariable {L}agrange
  inversion formula}, J. Combin. Theory Ser. A \textbf{45} (1987), no.~2,
  178--195.

\bibitem[Ges16]{gessel2016survey}
\bysame, \emph{Lagrange inversion}, J. Combin. Theory Ser. A \textbf{144}
  (2016), 212--249.

\bibitem[GK97]{goulden-kulkarni1997}
I.~P. Goulden and D.~M. Kulkarni, \emph{Multivariable {L}agrange inversion,
  {G}essel-{V}iennot cancellation, and the matrix tree theorem}, J. Combin.
  Theory Ser. A \textbf{80} (1997), no.~2, 295--308.

\bibitem[Goo60]{good1960}
I.~J. Good, \emph{Generalizations to several variables of {L}agrange's
  expansion, with applications to stochastic processes}, Proc. Cambridge
  Philos. Soc. \textbf{56} (1960), 367--380.

\bibitem[Goo65]{good1965}
\bysame, \emph{The generalization of {L}agrange's expansion and the enumeration
  of trees}, Proc. Cambridge Philos. Soc. \textbf{61} (1965), 499--517.

\bibitem[Hai18]{hairer2018bphz}
M.~Hairer, \emph{An analyst's take on the {BPHZ} theorem}, Computation and
  combinatorics in dynamics, stochastics and control, Abel Symp., vol.~13,
  Springer, Cham, 2018, pp.~429--476.

\bibitem[JKT19]{jansen-kuna-tsagkaro2019virialinversion}
S.~Jansen, T.~Kuna, and D.~Tsagkarogiannis, \emph{Virial inversion and density
  functionals}, Online preprintarXiv:1906.02322 [math-ph], 2019.

\bibitem[Joy81]{joyal1981}
A.~Joyal, \emph{Une th\'{e}orie combinatoire des s\'{e}ries formelles}, Adv. in
  Math. \textbf{42} (1981), no.~1, 1--82.

\bibitem[JP19]{johnson-prochno2019}
S.~G. Johnson and S.~Prochno, \emph{Trees and {F}a{\`a} die {B}runo's formula},
  Online preprint arXiv:1911.07458 [math.CO], 2019.

\bibitem[JTTU14]{jttu2014}
S.~Jansen, S.~J. Tate, D.~Tsagkarogiannis, and D.~Ueltschi, \emph{Multispecies
  virial expansions}, Comm. Math. Phys. \textbf{330} (2014), no.~2, 801--817.

\bibitem[Lab81]{labelle1981}
G.~Labelle, \emph{Une nouvelle d\'{e}monstration combinatoire des formules
  d'inversion de {L}agrange}, Adv. in Math. \textbf{42} (1981), no.~3,
  217--247.

\bibitem[MN93]{mendez-nava1993}
M.~M\'{e}ndez and O.~Nava, \emph{Colored species, {$c$}-monoids, and plethysm.
  {I}}, J. Combin. Theory Ser. A \textbf{64} (1993), no.~1, 102--129.

\bibitem[Per64]{percus1964lagrange}
J.~K. Percus, \emph{A note on extension of the {L}agrange inversion formula},
  Comm. Pure Appl. Math. \textbf{17} (1964), 137--146.

\bibitem[Pit06]{pitman2002book}
J.~Pitman, \emph{Combinatorial stochastic processes}, Lecture Notes in
  Mathematics, vol. 1875, Springer-Verlag, Berlin, 2006, Lectures from the 32nd
  Summer School on Probability Theory held in Saint-Flour, July 7--24, 2002,
  With a foreword by Jean Picard.

\bibitem[Rue69]{ruelle1969book}
D.~Ruelle, \emph{Statistical mechanics: Rigorous results}, World Scientific,
  1969.

\bibitem[Sim15]{Simon2015Operator}
Barry Simon, \emph{Operator theory}, A Comprehensive Course in Analysis, Part
  4, American Mathematical Society, Providence, RI, 2015. \MR{3364494}

\bibitem[Ste64]{stell1964}
G.~Stell, \emph{Cluster expansions for classical systems in equilibrium}, The
  Equilibrium Theory of Classical Fluids. (H.~L. Frisch and J.~L. Lebowitz,
  eds.), Benjamin, New York, 1964, pp.~171--261.

\bibitem[Wri89]{wright1989treeformula}
D.~Wright, \emph{The tree formulas for reversion of power series}, J. Pure
  Appl. Algebra \textbf{57} (1989), no.~2, 191--211.

\bibitem[Wri00]{wright2000}
\bysame, \emph{Reversion, trees, and the {J}acobian conjecture}, Combinatorial
  and computational algebra ({H}ong {K}ong, 1999), Contemp. Math., vol. 264,
  Amer. Math. Soc., Providence, RI, 2000, pp.~249--267.

\end{thebibliography}

\end{document}